\documentclass[11pt, a4paper]{article}
\pdfoutput=1 %

\usepackage{fullpage}

\usepackage{mdframed}

\usepackage[utf8]{inputenc}

\usepackage{libertine}

\usepackage{amsfonts}
\usepackage{amsmath}
\usepackage{amssymb}
\usepackage{amsthm}
\usepackage{float}
\usepackage{stmaryrd}

\usepackage{hyperref}
\usepackage[svgnames]{xcolor}
\hypersetup{colorlinks={true},urlcolor={blue},linkcolor={DarkBlue},citecolor=[named]{DarkGreen}}
\usepackage[nocompress]{cite}

\newif\iflong
\newif\ifshort

\longtrue

\iflong
\else
\shorttrue
\fi

\usepackage{xspace}
\usepackage{fancyhdr}
\usepackage{tcolorbox}

\usepackage{microtype}
\usepackage[capitalise,nameinlink]{cleveref}
\usepackage{enumitem}

\usepackage{cleveref}

\usepackage{doi}

\usepackage{authblk}

\usepackage{booktabs} %
\usepackage[ruled,linesnumbered]{algorithm2e} %
\usepackage{setspace}

\SetAlFnt{\small}
\SetAlCapFnt{\small}
\SetAlCapNameFnt{\small}
\SetAlCapHSkip{0pt}
\IncMargin{-\parindent}

\usepackage{tikz}  %
\usetikzlibrary{arrows}
\usetikzlibrary{patterns}
\usetikzlibrary{decorations.shapes}
\tikzstyle{overbrace text style}=[font=\tiny, above, pos=.5, yshift=5pt]
\tikzstyle{overbrace style}=[decorate,decoration={brace,raise=5pt,amplitude=3pt}]
\usetikzlibrary{shapes.geometric}

\usepackage{authblk}
\usepackage{xspace}

\theoremstyle{definition}
\newtheorem{definition}{Definition}

\theoremstyle{plain}
\newtheorem{theorem}{Theorem}[section]
\newtheorem{lemma}[theorem]{Lemma}
\newtheorem{corollary}[theorem]{Corollary}

\newtheorem{observation}{Observation}
\newtheorem{claim}{Claim}
\newenvironment{claimproof}[1]{\par\noindent\emph{Proof.}\hspace{0.15cm}#1}{\hfill $\blacktriangleleft$\smallskip}

\theoremstyle{definition}

\usepackage{amsfonts}
\usepackage{amssymb}
\usepackage{mathtools}
\usepackage{xspace}
\usepackage{xcolor}
\usepackage{cleveref}
\usepackage{todonotes}
\usepackage{comment}

\usepackage{boxedminipage}
\newcommand{\prob}[3]{
	\begin{center}
		\begin{boxedminipage}{\columnwidth}
			#1\\[5pt]
			\begin{tabular}{l p{0.77 \columnwidth}}
				{\bf\hspace{-0.2cm}Given:} & #2\\[5pt]
				{\bf\hspace{-0.2cm}Question:} \hspace{-0.4cm} & #3
			\end{tabular}
		\end{boxedminipage}
	\end{center}
}

\newcommand{\TDG}{\textsc{Topological Distance Game}\xspace}
\newcommand{\TDGshort}{\textsc{TDG}\xspace}
\newcommand{\IRTDG}{\textsc{IR}-\TDG}
\newcommand{\IRTDGshort}{\textsc{IR-}\TDGshort}

\newcommand{\agents}{\ensuremath{N}} %
\newcommand{\numagents}{\ensuremath{|\agents|}} %
\newcommand{\util}{\ensuremath{\operatorname{u}}} %
\newcommand{\assgn}{\ensuremath{\lambda}} %
\newcommand{\dffn}{\ensuremath{f}} %
\newcommand{\dist}{\ensuremath{\operatorname{dist}}} %

\newcommand{\sd}{\ensuremath{\operatorname{sd}}} %
\newcommand{\tww}{\ensuremath{\operatorname{tww}}} %

\newcommand{\NP}{\textsf{NP}\xspace}
\newcommand{\NPh}{\NP-hard\xspace}
\newcommand{\NPhness}{\NP-hardness\xspace}
\newcommand{\NPc}{\NP-complete\xspace}
\newcommand{\NPcness}{\NP-completeness\xspace}

\newcommand{\FPT}{\textsf{FPT}\xspace}
\newcommand{\XP}{\textsf{XP}\xspace}
\newcommand{\W}[1][1]{\textsf{W[#1]}\xspace}
\newcommand{\Wh}[1][1]{\W[#1]-hard\xspace}
\newcommand{\Whness}[1][1]{\W[#1]-hardness\xspace}
\newcommand{\Wc}[1][1]{\W[#1]-complete\xspace}
\newcommand{\pNPh}{\textsf{para-NP}-hard\xspace}

\newcommand{\Oh}[1]{\ensuremath{{\mathcal{O}\left(#1\right)}}}

\allowdisplaybreaks

\title{\bf Individual Rationality in Topological Distance Games is Surprisingly Hard}

\author[1]{Argyrios Deligkas}
\author[1]{Eduard Eiben}
\author[2]{Dušan Knop}
\author[2]{Šimon Schierreich}

\affil[1]{Royal Holloway, University of London, United Kingdom}
\affil[2]{Czech Technical University in Prague, Czech Republic}

\date{\today}

\begin{document}
	
\maketitle
\thispagestyle{empty}

\begin{abstract}
	In the recently introduced \emph{topological distance games}, strategic agents need to be assigned to a subset of vertices of a topology. In the assignment, the utility of an agent depends on both the agent's inherent utilities for other agents and its distance from them on the topology.
	We study the computational complexity of finding \emph{individually rational} outcomes; this notion is widely assumed to be the very minimal stability requirement and requires that the utility of every agent in a solution is non-negative.
	We perform a comprehensive study of the problem's complexity, and we prove that even in very basic cases, deciding whether an individually rational solution exists is intractable. To reach at least some tractability, one needs to combine multiple restrictions of the input instance, including the number of agents \emph{and} the topology \emph{and} the influence of distant agents on the utility.
\end{abstract}

\section{Introduction}
You are the coordinator of the annual banquet of your organization and your task is to convince all employees to attend the event. 
Clearly, a person agrees to show up at such an event, only if they get {\em at least} some positive experience from their participation.
However, the enmity between grumpy-John, prickly-Jack, and grouchy-Joe is known to everyone. 
It should be fairly easy to convince all three of them to attend the banquet if their seats are far away from each other and some friendly people sit between them. Right? As we will see, it is not easy at all.

Situations like the above occur in several other scenarios; think of assigning desks to students, offices to academics, or seats on an assembly line. In such cases, the happiness of a participant depends not only on who are their immediate neighbors and how close their friends are, but also on how far their ``enemies'' are located.
Recently, Bullinger and Suksompong \cite{BullingerS2023} proposed the elegant framework of {\em topological distance games} in order to model such preferences for the agents. 
In such a game, there is an underlying \emph{topology}, represented by an undirected graph, where a set of agents needs to be assigned on (a subset of) its vertices. 
Crucially, though, the {\em utility} of an agent depends not only on its inherent utility for other agents but also on the distance from them on the topology.

In this model, Bullinger and Suksompong \cite{BullingerS2023} studied the existence and the complexity of {\em stable} outcomes. More specifically, they have focused on {\em jump stability}, i.e., an assignment where no agent has incentives to ``jump'' to an {\em empty} vertex in order to increase their utility. 
However, they have assumed that the agents actually {\em want} to participate in the game, even if there is no way to receive positive utility. For example, imagine a topology with the same number of vertices as the number of agents and two agents that hate each other. Then, although every assignment is jump-stable, arguably these agents would not participate if they had to sit next to each other; it is simply not {\em individually rational}.

\subsection{Our Contribution}
We perform a comprehensive and in-depth study of the complexity of individual rationality (IR) in topological distance games with the aim of identifying the precise cut-off between tractable and intractable classes of instances.
IR is {\em ``a minimal requirement for a solution to be considered stable''}~\cite{AzizS16} and formally states that there exists an assignment that guarantees non-negative utility to every agent.

We identify several dimensions of the model -- the number of agents, the enmity graph, the distance factor function, and the topology structure -- and we sketch the complexity of the problem with respect to them. Here, the enmity graph is a directed graph that shows which agents are ``enemies'', i.e., get negative utility from their interaction, and the distance factor function is a monotonically decreasing function that weighs the utility of the agents depending on their distance.

We begin in Section~\ref{sec:agents_input} by considering the number of agents as part of the input and we show that ensuring IR in this case is extremely hard, even for very restricted cases; see \Cref{fig:results:agentsInput} for a simplified overview of the results of this section. 
We start our investigation by not imposing any restrictions on the enmity graph and we show that the problem is \NPc for {\em every} distance factor function, even when the utilities of the agents are symmetric and they have at most 2 different values per agent (Thm.~\ref{thm:unrestricted_ennmity}).
Hence, in order to hope for tractability, we need to restrict the enmity graph. 
We then show that, if at most one agent has enemies, the problem can be solved in polynomial time (Thm.~\ref{thm:IRTDG:P:singleArcEnmityGraph}).
Unfortunately, this is the best possible, since the problem is \NPc when there are two arcs in the enmity graph, i.e., there are two agents with enemies (this establishes a dichotomy with respect to the number of arcs of the enmity graph).
More specifically, we provide two different reasons for hardness when there are two arcs in the enmity graph: Theorem~\ref{thm:IRTDG:NPc:enmityGraphWithTwoArcs} shows that the problem is hard even when the utilities are symmetric, while Theorem~\ref{thm:two_arcs_same} shows hardness for {\em any} distance factor function when the two arcs are towards the same vertex and there are two types of utilities. 
Finally, we show that restricting only the topology does not help us either, as the problem remains intractable even when the topology is a path and there are four types of utilities (Thm.~\ref{thm:NPc_path}).

With the above considerations in mind, in Section~\ref{sec:agents_parameter} we take the number of agents as a parameter; see \Cref{fig:results:agentsParam} for a summary of results of this section. 
Our first result shows that there is an easy \XP algorithm for the problem (Thm.~\ref{thm:XP_agents}). 
However, without any further restrictions, this result is ``tight'', as the problem is \Wc even with two types of symmetric utility functions (Thm.~\ref{thm:IRTDG:Wh:agents}) and, moreover, the running time of the algorithm cannot be substantially improved under the well-known Exponential-Time Hypothesis (ETH). 

Next, we show that restricting ``just'' the topology is not sufficient for tractability, as the problem remains \Wc even on path topology; this is our technically most involved result (Thm.~\ref{thm:Wh:path}). 
On the positive side though, the problem is in \FPT on path topology when the enmity graph has arcs only towards at most one agent (Thm.~\ref{thm:FPT_path}); this combination of structural restrictions seems necessary since the problem becomes \Wc under the enmity graph above for general topologies, even when there are only two types of utilities (Thm.~\ref{thm:Wc:two_types}).
Our last set of results contains two \FPT algorithms. 
Theorem~\ref{thm:FPT_cw} establishes fixed-parameter tractability for bounded distance factor functions parameterized by the number of agents plus the twin-width of the topology. 
Twin-width~\cite{BonnetKTW22} is a structural graph parameter which is more general than several well-studied parameters like tree-width and clique-width and is constant for several large families of graphs such as planar graphs~\cite{HlinenyJ2023}.
To the best of our knowledge, this is the first fixed-parameter algorithm for twin-width in the areas of algorithmic game theory and computational social choice.
Theorem~\ref{thm:FPT_cw} essentially yields fixed-parameter tractability for instances with an arbitrary distance factor function parameterized by the number of agents and the shrub-depth of the topology combined (Cor.~\ref{thm:FPT_sd}). Shrub-depth, informally speaking, corresponds to the smallest depth of a tree into which the topology can be embedded.

\subsection{Related Work}
Topological Distance Games are closely related to many well-known classes of coalition formation and network games.

The first important source of inspiration includes \emph{hedonic games}~\cite{DrezeG1980}, a prominent model in coalition formation. Here, we are given a set of agents together with their preferences, and our goal is to partition them into coalitions. The crucial property of hedonic games is that the agent's utility is based solely on other members of his or her coalition. In general hedonic games, every agent $a$ has preferences over possible coalitions (subsets of agents) containing $a$. It follows that such preferences cannot be represented succinctly, and therefore many variants with restricted preferences are studied, such as graphical~\cite{Peters2016a,HanakaL2022}, fractional~\cite{AzizBBHOP2019,FanelliMM2021}, anonymous~\cite{BogomolnaiaJ2002}, or diversity~\cite{BredereckEI2019,GanianHKSS2023,Darmann2023}. The variant that is closest to our setting is hedonic games with additively-separable preferences~\cite{BogomolnaiaJ2002,AzizBS2013} (ADHGs), where each agent $a$ assigns some value to each other agent $b$ and the utility for agent~$a$ is simply the sum of values agent $a$ has for all other agents in its coalition. The modeling of ADHGs in our model is very straightforward; the topology consists of $n$ cliques, each of size $n$ (or $k$ in the case of fixed-size coalitions~\cite{BiloMM2022,LiMNN2023}), where $n$ is the number of agents.
It should be noted that achieving individual rationality in ADHGs is trivial: we put each agent into its own coalition.

Closely related are also \emph{social distance games}~\cite{BranzeiL2011,KaklamanisKP2018,BalliuFMO2019,BalliuFMO2022}, where our goal is again to partition agents into coalitions. This time, the agents are given together with a topology representing relations between them. Agent's utility with respect to a coalition is then the average of the reciprocal distances to all other agents in this coalition; however, we assume the distances with respect to the subgraph induced by the members of this coalition. Consequently, the role of the topology in social distance games is very different compared to TDGs. Later, Flammini et al.~\cite{FlamminiKOV2020} generalized social distance games by adding a global scoring vector that allows us to extend the model beyond the reciprocal distance function. This direction was further developed in Ganian et al.~\cite{GanianHKRSO2023}, who studied the computational complexity of this generalization with respect to multiple stability notions, including individual rationality. 

None of the above-mentioned models included the assignment of agents to a topology. In this line of research, very prominent is Schelling's segregation model~\cite{Schelling1969,Schelling1971} and its game-theoretical refinement called \emph{Schelling games}~\cite{ChauhanLM2018,EchzellFLMPSSS2019,AgarwalEGISV2021,KreiselBFN2022,BiloBLM2022a,BiloBLM2022b,FriedrichLMS2023,DeligkasEG2023,BiloBDLM02023}. Here, we are given a set of agents and a topology, and our goal is to assign agents to the topology in a desirable way. However, in contrast to TDGs, in Schelling games, the agents are additionally partitioned into types, and the utility of each agent is implicitly derived from the fraction of agents of the same type assigned to its neighborhood. 

A similar situation, that is, agents' utilities are based solely on their neighbors, also appears in \emph{hedonic seat arrangement}~\cite{BodlaenderHJOOZ2020,Ceylan0R2023,Wilczynski2023}, where preferences can be more general, or recently introduced \emph{refugee housing}~\cite{KnopS2023,Schierreich2023,LisowskiS2023}, where we additionally have a subset of agents that are initially assigned to some vertices of the topology, and our task is to assign the remaining agents in a sort of IR manner.

\section{Preliminaries}

For each positive integer $i$, we define $[i]$ to be the set $\{1,\ldots,i\}$. For a set $S$ and a positive integer $k$, we denote by $\binom{S}{k}$ the set of all $k$-sized subsets of $S$, and by $2^S$ we denote the set of all subsets of $S$.

\subsection{Graph Theory} 
We follow the standard graph-theoretical notation~\cite{Diestel2017}. A simple undirected graph $G$ is a pair $(V,E)$, where $V$ is a non-empty set of vertices and $E\subseteq\binom{V}{2}$ is a set of edges. For two vertices $u,v\in V$, we denote by $\dist_G(u,v)$ the length of the shortest path between $u$ and $v$ in the graph $G$, and we set $\dist_G(u,v) = \infty$ if there is no $u,v$-path in $G$.

\subsection{Topological Distance Games}
We use $\agents$ to denote the set of agents.
Each agent $i\in\agents$ is accompanied with a \emph{utility function} $\util_i\colon\agents\to\mathbb{R}$ such that $\util_i(i) = 0$.
We say that agent~$j$ is a \emph{friend} of agent $i$ if $\util_i(j) > 0$. If $\util_i(j) < 0$, the agent~$j$ is an \emph{enemy} of the agent $i$. %
The \emph{enmity graph} is a directed graph on the set $\agents$ of agents such that there is an edge from an agent $i$ to an agent $j$ if and only if $j$ is an enemy of $i$.

The \emph{topology} is a simple undirected graph $G=(V,E)$ with at least $\numagents$ vertices. An \emph{assignment} is an injective mapping $\assgn\colon\agents\to V$ assigning agents to vertices of the topology. The \emph{distance factor function} $\dffn\colon \mathbb{N}\to \mathbb{R}_{> 0}$ is a strictly decreasing function that scales the influence of one agent to another agent based on their distance in the topology. In addition, we define $\dffn(\infty) = 0$.
We further extend the utility function for assignments as follows. Given an assignment $\assgn$, we define its utility $\util_i(\assgn)$ as
\[
\util_i(\assgn) = \sum_{j\in\agents\setminus\{i\}} \util_i(j)\cdot \dffn(\dist_G(\assgn(i),\assgn(j))).
\]

\begin{definition}
	An assignment $\assgn$ is called \emph{individually rational} if for every agent $i\in\agents$ we have $\util_i(\assgn) \geq 0$.
\end{definition}

Now, we are ready to formally define the computational problem of our interest.

\prob{\IRTDG (\IRTDGshort for short)}
{A topology $G$, a set of agents $\agents$, a utility function $\util_i$ for every agent $i\in\agents$, and a distance factor function~$\dffn$.}
{Is there an assignment $\assgn$ that is individually rational?}

\subsection{Parameterized Complexity} 
The framework of parameterized complexity~\cite{Niedermeier2006,DowneyF2013,CyganFKLMPPS2015} gives us formal tools for a finer-grained complexity of computational problems that are assumed to be intractable in their full generality. Informally, under this problem, we study variants of intractable problems that are somehow restricted, and this restriction is captured in the so-called \emph{parameter}~$k$. The ultimate goal is then to invent algorithms such that the exponential blow-up in the running time can be confined to the parameter and not to the input size. 
In this direction, the best possible outcome is an algorithm running in $g(k)\cdot n^\Oh{1}$ time for any computable function $g$. Such an algorithm is called \emph{fixed-parameter algorithm}, and \FPT is the class of all parameterized problems that admit a fixed-parameter algorithm. Slightly worse, but still positive, is an algorithm running in $g(k)\cdot n^{h(k)}$ time, where~$g,h$ are computable functions. The complexity class containing all parameterized problems admitting such algorithms is called \XP. One can rule out the existence of a fixed-parameter algorithm by proving that the problem of interest is \Wh. This can be shown by a parameterized reduction from some other \Wh parameterized problem; \textsc{Independent Set} parameterized by the solution size is a prototypical \Wh problem.
Similarly, one can exclude the existence even of an \XP algorithm by showing that the parameterized problem is \NPh already for a constant value of the parameter. Such a problem is then called \pNPh.
For a more comprehensive introduction to the parameterized complexity, we refer the interested reader to the monograph of~\cite{CyganFKLMPPS2015}.

\paragraph{Exponential-Time Hypothesis.} 
The \emph{Exponential-Time Hypothesis} (ETH) of Impagliazzo et al.~\cite{ImpagliazzoP2001,ImpagliazzoPZ2001} is a well-established assumption in theoretical computer science that, informally speaking, states that there cannot exist an algorithm for the $3$-SAT problem with running time which is sub-exponential in the number of variables.

\section{Unrestricted number of agents}
\label{sec:agents_input}

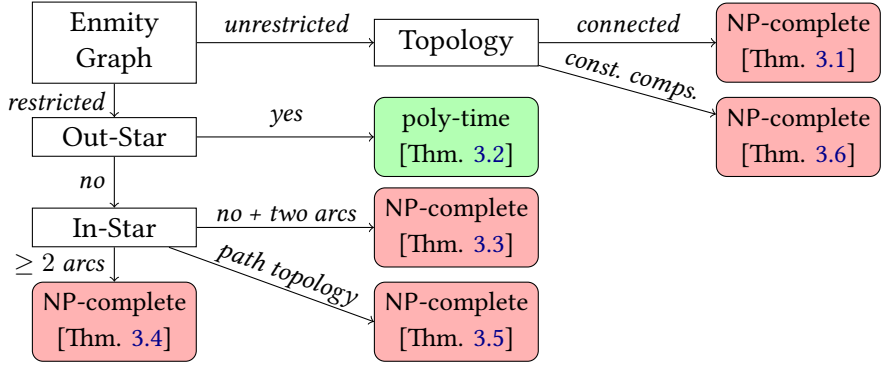
\begin{figure}
	\centering
	\begin{tikzpicture}[align=center,node distance=1.2cm]
		\tikzstyle{decision} = [draw,align=center,text width=1.9cm];
		\tikzstyle{result} = [draw,rectangle,rounded corners,align=center,text width=1.9cm,node distance=4.5cm];
		\tikzstyle{NPh} = [fill=red!30];
		\tikzstyle{poly} = [fill=green!30];
		
		\node[decision] (d1) at (0,0) {Enmity Graph};
		\node[decision,below of=d1,node distance=1.25cm] (d2) {Out-Star};
		\node[decision,below of=d2] (d3) {In-Star};
		
		\node[decision,right of=d1,node distance=4.5cm] (d4) {Topology};
		\node[result,NPh,right of=d4] (p1) {\small\NPc\\\small[Thm.~\ref{thm:unrestricted_ennmity}]};
		
		\node[result,poly,right of=d2] (p2) {\small poly-time\\\small[Thm.~\ref{thm:IRTDG:P:singleArcEnmityGraph}]};
		\node[result,NPh,right of=p2] (p6) {\small\NPc\\\small[Thm.~\ref{thm:const_components}]};
		
		\node[result,NPh,right of=d3] (p5) {\small\NPc\\\small[Thm.~\ref{thm:IRTDG:NPc:enmityGraphWithTwoArcs}]};
		\node[result,NPh,below of=d3,node distance=1.25cm] (p3) {\small\NPc\\\small[Thm.~\ref{thm:two_arcs_same}]};
		\node[result,NPh,below of=p5,node distance=1.25cm] (p4) {\small\NPc\\\small[Thm.~\ref{thm:NPc_path}]};
		
		\draw[->] (d1) -- (d4) node [above,midway] {\it\small unrestricted};
		\draw[->] (d4) -- (p1) node [above,midway] {\it\small connected};
		\draw[->] (d4) -- (p6) node [above,midway,sloped,yshift=-3pt] {\it\small const. comps.};
		\draw[->] (d1) -- (d2) node [left,midway] {\it\small restricted};
		\draw[->] (d2) -- (p2) node [above,midway] {\it\small yes};
		\draw[->] (d2) -- (d3) node [left,midway] {\it\small no};
		\draw[->] (d3) -- (p3) node [left,midway] {\it\small$\geq 2 \text{ arcs}$};
		\draw[->] (d3) -- (p4.west) node [above,midway,sloped,xshift=0.15cm,yshift=-0.1cm] {\it\small path topology};
		\draw[->] (d3) -- (p5) node [above,midway,yshift=-0.05cm] {\it\small no + two arcs};
	\end{tikzpicture}
	\caption{A simplified overview of our results when the number of agents is part of the input. By \emph{const. comps.} we mean that the topology is disconnected and consists of constant-size components.}
	\label{fig:results:agentsInput}
\end{figure}

As our first result, we show that it is easy to verify whether a given assignment $\assgn$ is individually rational or not.

\begin{observation}\label{obs:IRTDG:NP}
	The \IRTDG problem is in \NP.
\end{observation}
\begin{proof}
	Given an assignment $\assgn$ and an agent $i\in\agents$, we can compute in polynomial time the utility $\util_i(\assgn)$ and check whether this value is non-negative. As there are finitely many agents, the problem is clearly in \NP.
\end{proof}

From this point onward, we will not discuss the belonging to \NP separately in any of our \NPcness proof.

In our first negative result we show that the problem is intractable even if we severely restrict the utilities of the agents.

\begin{theorem}\label{thm:unrestricted_ennmity}
	For every distance factor function $\dffn$, it is \NPc to decide the \IRTDG problem even if the utilities are symmetric and every agent uses at most $2$ different utility values.
\end{theorem}
\begin{proof}
	We provide a reduction from the \textsc{Unary Bin Packing} problem~\cite{GareyJ1979}. The input of this problem is a list $S=(s_1,\ldots,s_n)$ of positive integers given in unary, the number of bins $B$, and a capacity $c$ of bins. The question is then whether there exist an allocation $\alpha\colon S\to [B]$ such that $\forall j\in[B]\colon \sum_{i\in[n]\colon \alpha(s_i) = j} s_i = c$. 
	The problem remains \NPc even if $\forall i \in [n]\colon s_i > 1$, $\sum_{i\in[n]} s_i = B\cdot c$.
	
	Given an instance $\mathcal{I} = (S,B,c)$ of \textsc{Unary Bin Packing}, we construct an equivalent instance $\mathcal{J}$ of the \IRTDG problem as follows. We start with the topology $G$, which is a disjoint union of~$B$ cliques $C_1,\ldots,C_B$, each of size~$c$. Since all vertices are in distance either one or infinity, the distance factor function $\dffn$ can be arbitrary. For the sake of exposition, we assume that $\dffn(1) = 1$. Next, we define the agents and the utilities. For every item $s_i \in S$, we create $s_i$ agents $a_{i,1},\ldots,a_{i,s_i}$. The utility function of these agents is the same and is constructed such that these agents have to be part of the same clique; otherwise, their utility is necessarily negative. Specifically, we set $\util_{a_{i,j}}(a_{\ell,k}) = -1$, where $i\in[n]$, $j\in[s_i]$, $\ell\in[n]\setminus i$, and $k\in[s_\ell]$, and $\util_{a_{i,j}}(a_{i,\ell}) = \frac{c-s_i}{s_i-1}$, where $i\in[n]$ and $j,\ell\in[s_i]$.
	Observe that the number of vertices of $G$ and the number of agents is the same. 
	The utilities are indeed symmetric, and every agent uses $2$ different values in the utility function.
	
	For correctness, assume first that $\mathcal{I}$ is a \emph{yes}-instance and $\alpha$ is a correct allocation of items. We create an assignment $\assgn$ as follows. For every item $s_i\in S$, we assign the corresponding agents $a_{i,1},\ldots,a_{i,s_i}$ arbitrarily to empty vertices of the clique $C_{\alpha(s_i)}$. Since $\alpha$ is a valid allocation, there are enough empty vertices for each agent. Also, observe that all agents corresponding to a single item $s_i$ are assigned to the same clique, and therefore, for every $j\in[s_i]$, it holds that
	\begin{align*}
		\util_{a_{i,j}}(\assgn) &= (c-s_i)\cdot(-1) + (s_i - 1)\cdot\frac{(c-s_i)}{(s_i-1)}\\
		&= -c + s_i + c - s_i = 0.
	\end{align*}
	Hence, the utility of every agent is $0$ with respect to $\assgn$. Therefore, the assignment $\assgn$ is individually rational, and $\mathcal{J}$ is also a \emph{yes}-instance.
	
	In the opposite direction, let $\mathcal{J}$ be a \emph{yes}-instance and $\assgn$ be an individually rational assignment. First of all, we prove an auxiliary claim stating that all agents corresponding to the same item are in the same clique.
	
	\begin{claim}\label{clm:IRTDG:NPc:symmetricAtMostTwoPerAgent:allInSameComp}
		Let $a_{i,j}$ and $a_{i,\ell}$, where $i\in[n]$ and $j,\ell\in[s_i]$ are distinct, be two agents. In every individually rational assignment $\assgn$, it holds that $\dist(\assgn(a_{i,j}),\assgn(a_{i,\ell})) = 1$.
	\end{claim}
	\begin{claimproof}
		For the sake of contradiction, assume that $\assgn$ is individually rational and $\dist(\assgn(a_{i,j}),\assgn(a_{i,\ell})) \not= 1$. Since all the components of~$G$ are cliques, it follows that $\dist(\assgn(a_{i,j}),\assgn(a_{i,\ell})) = -\infty$. Consequently, we have 
		\begin{align*}
			\util_{a_{i,j}}(\assgn) &\leq (c- s_i + 1)\cdot (-1) + (s_i - 2)\cdot\frac{(c-s_i)}{(s_i-1)} \\
			&= -c\cdot s_i + s_i^2 - s_i + (s_i - 2)\cdot(c-s_i) \\
			&= -c\cdot s_i + s_i^2 - s_i + c\cdot s_i - s_i^2 - 2c + 2s_i \\
			&= -2c + s_i \leq -c,
		\end{align*}
		and therefore, such $\assgn$ is not individually rational. Thus, in every $\assgn$, all $a_{i,j}$ and $a_{i,\ell}$ are in distance one.
	\end{claimproof}
	
	With \Cref{clm:IRTDG:NPc:symmetricAtMostTwoPerAgent:allInSameComp} in hand, we can directly create a valid allocation $\alpha$ for $\mathcal{I}$. In particular, we set $\alpha(s_i)$ equal to the index of the clique the agent $a_{i,1}$ is assigned to. Clearly, all items are allocated. Additionally, assume that there is a bin $j\in[B]$ such that $\sum_{i\in[n]\colon \alpha(s_i) = j} s_i < c$. Then for $K_j$ it holds that $\sum_{i\in[n]\colon\assgn(a_{i,1})\in V(K_j)} s_i < c$, which is, however, not possible as all vertices are occupied and, by \Cref{clm:IRTDG:NPc:symmetricAtMostTwoPerAgent:allInSameComp}, all agents corresponding to $s_i$ are part of the same clique. Therefore, $\alpha$ is solution for $\mathcal{I}$.
	
	The reduction can be clearly done in polynomial time. Together with \Cref{obs:IRTDG:NP}, we obtain that \IRTDGshort is \NPc.
\end{proof}

Then, we prove that the problem is tractable when there is at most one agent that has enemies.

\begin{theorem}\label{thm:IRTDG:P:singleArcEnmityGraph}
	If there is at most one agent $p$ assigning negative utility to other agents, the \IRTDG problem can be solved in polynomial time for any distance factor function $\dffn$.
\end{theorem}
\begin{proof}
	If there is no arc in the enmity graph, then all the utilities are positive, and therefore, every assignment is individually rational. Thus, let there be at least one arc in the enmity graph. We can split the agents into two sets $\agents^+$ and $\agents^-$ according to the utility the agent $p$ has for them. 
	Formally, we set $\agents^+ = \{i\in\agents\setminus\{p\}\mid \util_p(i) \geq 0\}$ and $\agents^- = \{i\in\agents\setminus\{p\}\mid \util_p(i) < 0\}$, respectively.
	Then, we try all possible assignments of agent~$p$ to the topology, and for every possibility, we do the following. Let~$v$ be the vertex the agent~$p$ is assigned in the currently examined possibility. We run the Breadth-first search algorithm starting with the vertex $v$ to find a BFS-tree~$T$. Now, we do a level order traversal of the tree~$T$, and for each vertex~$u$ of~$T$, we assign to~$u$ an agent $i\in\agents^+$ that was not assigned before, and the agent~$p$ has for it the highest utility between all agents in $\agents^+$. As the final step, we assign the agents from~$\agents^-$. This is again done by the level-order traversal with the following differences. The traversal is done from the deepest level, and the agents are assigned according to the increasing utility that the agent $p$ has for them. If, for this assignment, the utility of~$i$ is non-negative, we return \emph{yes} and exit the algorithm. Otherwise, we will continue with another possibility. If no possibility leads to an individually rational assignment, we return \emph{no}.
	
	For correctness, we prove that if there is at least one individually rational assignment, there is also a solution of the form that our algorithm checks. Let $\assgn$ be an IR assignment and assume that~$j\in\agents^-$ such that $\util_p(j)$ is minimum is not in the farthest possible vertex, say $u$, from $\assgn(p)$. Then, by moving $j$ from $\assgn(j)$ to $u$, we can even increase the value of $\util_p(\assgn)$ since the distance factor function $\dffn$ is strictly decreasing. Similarly, it can be shown that by moving agents for which the agent~$i$ has positive utility, we can only increase the utility of~$p$ in~$\assgn$. Therefore, our algorithm exhaustively tries all relevant solutions and is clearly correct.
	
	Overall, we try $\Oh{|V(G)|}$ possible positions of~$p$, and for every position, we run BFS and assign the remaining agents to the graph. Therefore, the overall running time is~$\Oh{|V(G)^3|}$.
\end{proof}

In our next result, we show that a single arc in the enmity graph (cf.~\Cref{thm:IRTDG:P:singleArcEnmityGraph}) is basically the only restriction that makes the problem tractable. Specifically, in our next result, we show that if there are two arcs in the enmity graph, the problem becomes intractable.

The \NPhness is proven via a reduction from the \textsc{Equitable Partition} problem~\cite{GareyJ1979}. 
In fact, we start with this problem in Theorems~\ref{thm:two_arcs_same} and~\ref{thm:NPc_path} as well.
In this problem, we are given a list $S=(s_1,\ldots,s_{2n})$ of $2n$ positive integers such that $\sum_{i\in[2n]} s_i = 2k$, and the goal is to decide whether there exists a set $I\subseteq[2n]$ of size $n$ such that $\sum_{i\in I} s_i = \sum_{i\in [2n]\setminus I} s_i = k$. Without loss of generality, we can assume that $\min S \geq n^2$ and that for any $i,j\in [2n]$ we have $|s_i - s_j|\le \frac{\min S}{n^2}$~\cite{DeligkasEKS2024}.
In particular, this means that for any $J\subseteq [2n]$ with $|J|<n$, we have $\sum_{i\in J}s_i < k$.

\begin{theorem}\label{thm:IRTDG:NPc:enmityGraphWithTwoArcs}
	For every distance factor function $\dffn$, it is \NPc to decide the \IRTDG problem even if the enmity graph contains only two arcs and the utilities are symmetric. 
\end{theorem}
\begin{proof}
	Given an instance $S$ of the \textsc{Equitable Partition} problem, we construct an equivalent instance $\mathcal{J}$ of the \IRTDG problem as follows. First, we construct the topology $G$. At the beginning, we create a complete bipartite graph $K_{n,n}$ with two parts $L$ and $R$. Then, we add a vertex $v_\ell$, which is connected with all vertices of the part $L$, and a vertex $v_r$, which is connected with all vertices of the part $R$. The set of agents consists of $2n$ \emph{element-agents}, each corresponding to one element of the set $S$, and two \emph{guard-agents} $g_1$, $g_2$. The idea behind the construction is that the guards hate each other, and the only way to make their utility non-negative is to assign to vertices~$v_\ell$ and~$v_r$, respectively, and to partition the element-agents between two parts of $K_{n,n}$ such that utility the agents~$g_1$ and $g_2$ gain from neighboring agents is exactly~$k$. To ensure this, we define the utilities as follows. Let $\dffn$ be an arbitrary but fixed distance factor function. For the guard-agents, we set $\util_{g_1}(g_2) = \util_{g_2}(g_1) = -(k + \frac{\dffn(2)}{\dffn(1)}\cdot k)/\dffn(2)$. Next, let $s_i$, $i\in[2n]$, be an element-agent. We set $\util_{g_1}(a_i) = \util_{g_2}(a_i) = \util_{a_i}(g_1) = \util_{a_i}(g_2) = \frac{s_i}{\dffn(1)}$. The remaining utilities, that is, between element-agents, are zero.
	
	For correctness, let $\mathcal{I}$ be a \emph{yes}-instance and $I\subseteq[2n]$ be a solution partition. We construct a solution assignment $\assgn$ for~$\mathcal{J}$ as follows. First, we set $\assgn(g_1) = v_\ell$ and $\assgn(g_2) = v_r$. Next, for every $i\in I$, we assign the element-agent $a_i$ to an empty vertex of $L$, and for every $i\in[2n]\setminus I$, we assign $a_i$ to an empty vertex of $R$. Since $I$ is an equitable partition, that is, $|I| = n$, there is always an empty vertex for each element-agent. Moreover, we have that
	\begin{align*}
		\util_{g_1}(\assgn) &= 
		\sum_{i\in I} \dffn(1)\cdot \frac{s_i}{\dffn(1)} 
		+ \sum_{\mathclap{i\in[2n]\setminus I}} \dffn(2)\cdot \frac{s_i}{\dffn(1)} 
		+ \dffn(3)\cdot \util_{g_1}(g_2)\\
		&= k + \frac{\dffn(2)}{\dffn(1)}\cdot k + \dffn(2)\cdot \util_{g_1}(g_2)\\
		&= k + \frac{\dffn(2)}{\dffn(1)}\cdot k 
		+ \dffn(2)\cdot \frac{-(k + \frac{\dffn(2)}{\dffn(1)}\cdot k)}{\dffn(3)}\\
		&= k + \frac{\dffn(2)}{\dffn(1)}\cdot k - k - \frac{\dffn(2)}{\dffn(1)}\cdot k = 0.
	\end{align*}
	And the same holds for the guard-agent $g_2$ by symmetric arguments. Moreover, the utility of every element-agent is clearly non-negative as they get non-negative values from all remaining agents. Therefore, $\assgn$ is individually rational and $\mathcal{J}$ is also a \emph{yes}-instance.
	
	In the opposite direction, let $\mathcal{J}$ be a \emph{yes}-instance and $\assgn$ be an individually rational assignment. First of all, we show that the guard agents $g_1$ and $g_2$ are necessarily at a distance $3$, as otherwise, their utility would be negative.
	
	\begin{claim}
		For every individually rational assignment $\assgn$, we have $\dist(\assgn(g_1),\assgn(g_2)) = 3$.
	\end{claim}
	\begin{claimproof}
		For the sake of contradiction, let $\assgn$ be individually rational and assume that $\dist(\assgn(g_1),\assgn(g_2)) = 1$. Then the utility of agent $g_1$ (and by same arguments also of $g_2$) is at most
		\begin{align*}
			\util_{g_1}(\assgn) &= \dffn(1)\cdot\sum_{\substack{i\in [2n]\\\dist(\assgn(g_1),\assgn(a_i))=1}} \frac{s_i}{\dffn(1)} + \dffn(2)\cdot\sum_{\substack{i\in [2n]\\\dist(\assgn(g_1),\assgn(a_i))=2}} \frac{s_i}{\dffn(1)} + \dffn(1)\cdot\util_{g_1}(g_2)\\
			&< k+\max S + \frac{\dffn(2)}{\dffn(1)}\cdot(k-\max S) + \dffn(1)\cdot\util_{g_1}(g_2)\\
			&= k+\max S + \frac{\dffn(2)}{\dffn(1)}\cdot k - \frac{\dffn(2)}{\dffn(1)}\cdot\max S-\frac{\dffn(1)}{\dffn(3)}\cdot k - \frac{\dffn(2)}{\dffn(3)}\cdot k\\
			&= k\cdot\left(\frac{\dffn(1)\dffn(3)+\dffn(2)\dffn(3)-\dffn(1)^2-\dffn(1)\dffn(2)}{\dffn(1)\dffn(3)}\right) + \max S\cdot\left(\frac{\dffn(1)-\dffn(2)}{\dffn(1)}\right)\\
			&= k\cdot \alpha + \max S\cdot \beta,
		\end{align*}
		which is clearly negative, as it holds that $\alpha < 0$ and $|\alpha| > \beta$. Hence, the distance between $g_1$ and $g_2$ is at least $2$. Assume now that the distance is exactly two. But then, the utility of guard-agent with smaller (wlog let $\util_{g_1}(\assgn)\leq\util_{g_2}(\assgn)$) utility is at most
		\begin{align*}
			\util_{g_1}(\assgn) &= \dffn(1)\cdot\sum_{\substack{i\in [2n]\\\dist(\assgn(g_1),\assgn(a_i))=2}} \frac{s_i}{\dffn(1)} + \dffn(2)\cdot\sum_{\substack{i\in [2n]\\\dist(\assgn(g_1),\assgn(a_i))=2}} \frac{s_i}{\dffn(1)} + \dffn(2)\cdot\util_{g_1}(g_2)\\
			&< k + 2\max S + \frac{\dffn(2)}{\dffn(1)}\cdot(k-2\max S) + \dffn(2)\cdot \util_{g_1}(g_2)\\
			&= k + 2\max S + \frac{\dffn(2)}{\dffn(1)}\cdot k-\frac{\dffn(2)}{\dffn(1)}\cdot2\max S -\frac{\dffn(2)}{\dffn(3)}\cdot k - \frac{\dffn(2)^2}{\dffn(1)\dffn(3)}\cdot k\\
			&= k\cdot\left(\frac{\dffn(1)\dffn(3)+\dffn(2)\dffn(3)-\dffn(1)\dffn(2)-\dffn(2)\dffn(2)}{\dffn(1)\dffn(3)}\right) + \max S\cdot\left(\frac{2\dffn(1)-2\dffn(2)}{\dffn(2)}\right)\\
			&= k\cdot\alpha + \max S\cdot\beta,
		\end{align*}
		which is, again, strictly negative, and $\assgn$ is not individually rational, which is contradiction.
	\end{claimproof}
	
	Therefore, according to the previous claim, the distance between $g_1$ and $g_2$ is exactly $3$, and they must be allocated to~$v_\ell$ and $v_r$, respectively.
	
	Now, we create a solution $I$ for $\mathcal{I}$. Specifically, we set $I = \{i\in[2n]\mid \assgn(a_i)\in V(L)\}$. From the shape of the topology, it holds that $|I|=n$. Suppose that $\sum_{i\in I} s_i \not= \sum_{i \in [2n]\setminus I} s_i$ and without loss of generality, let the left part of the comparison be smaller and $g_1$ be assigned to $v_\ell$ according to $\assgn$. Then the utility of $g_1$ was at most
	\begin{align*}
		\util_{g_1}(\assgn) &\leq k - \min S + \frac{\dffn(2)}{\dffn(1)}\cdot(k+\min S)- k - \frac{\dffn(2)}{\dffn(1)}\cdot k\\
		&= \min S\cdot\left(\frac{\dffn(2)}{\dffn(1)}-1\right) = \min S\cdot\left(\frac{\dffn(2)-\dffn(1)}{\dffn(1)}\right) < 0,
	\end{align*}
	and therefore, $\assgn$ was not individually rational, which is a contradiction. Consequently, the sums on both sides are necessarily the same and $I$ is indeed solution for $\mathcal{I}$.
\end{proof}

Next, we show that the problem remains hard even if the enmity graph consists of two arcs pointed towards the same agent; in other words, if we ignore isolated vertices, the enmity graph is an in-star with two arcs.

\begin{theorem}\label{thm:two_arcs_same}
	For any distance factor function $f$, it is \NPc to decide the \IRTDG problem, even if there are only two arcs in the enmity graph and both of them are directed towards the same agent.
\end{theorem}
\begin{proof}
	We will prove our result via a reduction again from \textsc{Equitable Partition}. So, given an instance $S$ of equitable partition with $2n$ integers, we construct an equivalent instance $\mathcal{J}$ of the \IRTDG problem as follows.
	There will be $2n+3$ agents $a_1, a_2, \ldots, a_{2n}, h_1, h_2, b$. 
	The enmity graph consists of two arcs directed from $h_1$ and $h_2$ towards $b$.
	The topology has $2n+3$ vertices and is depicted at Figure~\ref{fig:two_arcs_same}; for $i \in \{1,2\}$ there exist $n$ vertices between vertex $H_i$ and $B$. 
	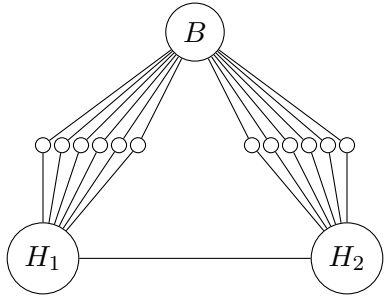
\begin{figure}
		\centering
		\begin{tikzpicture}
			\node[draw,circle] (b) at (0,0) {$B$};
			\node[draw,circle] (h1) at (-2,-3) {$H_1$};
			\node[draw,circle] (h2) at (2,-3) {$H_2$};
			
			\foreach[count=\i] \x in {-2,-1.75,-1.5,...,-0.75} {
				\node[draw,circle,inner sep=2pt] (v\i) at (\x,-1.5) {};
				\draw (v\i) edge (b) edge (h1);
			}
			
			\foreach[count=\i] \x in {2,1.75,1.5,...,0.75} {
				\node[draw,circle,inner sep=2pt] (v\i) at (\x,-1.5) {};
				\draw (v\i) edge (b) edge (h2);
			}
			
			\draw (h1) -- (h2);
		\end{tikzpicture}
		\caption{The topology used in Theorem~\ref{thm:two_arcs_same}.}
		\label{fig:two_arcs_same}
	\end{figure}
	Observe that since the diameter of the topology is 2, we need to consider only values $f(1)$ and $f(2)$ of the distance factor function $f$.
	There are two types of utility functions. Agents $h_1$ and $h_2$ have the same type of utility and for every $i \in \{1,2\}$ we have: $u_{h_i}(a_j)=s_j$ for every $j \in [2n]$; $u_{h_i}(b)=\frac{f(1)+f(2)}{f(2)}\cdot k$; and 0 otherwise. The rest of the agents get utility 0 from every other agent.
	
	The high-level idea is that if an individually rational solution exists, then agents $h_1$ and $h_2$ occupy vertices $H_1$ and $H_2$ respectively, agent $b$ occupies vertex $B$, and the remaining agents are split to ``left'' and ``right'' such that they correspond to an equitable partition.
	
	Assume that we have a solution $(S_1, S_2)$ for the given instance of \textsc{Equitable Partition}, i.e., for every $i \in \{1,2\}$ it holds that $\sum_{j \in S_i}s_i=k$. Then we create the following allocation for the agents on the topology. Agent $b$ is located on $B$, and for $i \in \{1,2\}$ $h_1$ agent $h_i$ is located on $H_i$ and we place agents from $S_i$ on the vertices between $B$ and $H_i$. Individual rationality depends only on the utility of the agents $h_1$ and $h_2$; everyone else gets utility 0 independently from how the agents are placed. Indeed, observe that the utility of agent $h_i$ is $f(1)\cdot \sum_{j \in S_i}s_i + f(2) \cdot \sum_{j \in S_i}s_i - f(2) \cdot \frac{f(1)+f(2)}{f(2)}\cdot k = f(1)\cdot k + f(2) \cdot k - f(2) \cdot \frac{f(1)+f(2)}{f(2)}\cdot k = 0$. Hence the allocation is individually rational.
	
	Conversely, assume that we have an individually rational solution for the instance we have created. 
	We argue that agent $b$ must be located on $B$, and for $i \in \{1,2\}$ agent $h_i$ must be located on $H_i$. 
	Observe that in any solution the distance between agent $b$ and agents $h_1$ and $h_2$ must be exactly two (since it cannot be more than two), otherwise the utility of at least one of $h_1$ and $h_2$ is strictly negative.
	Hence, each one of $h_1$ and $h_2$ gets utility $-(f(1)+f(2))\cdot k$ due to agent $b$. 
	Thus, in order both to get non-negative utility, each one of them should get utility at least $(f(1)+f(2))\cdot k$ from the other agents. It is not hard to see that this is possible only when $b$ is located on $B$, and agent $h_i$ must be located on $H_i$; in any other case, one of $h_1$ and $h_2$ will gain strictly less utility. Additionally, due to the topology structure, observe that in order both agents gain $(f(1)+f(2))\cdot k$ from the remaining agents, it means that the utility from the $n$ agents that are in distance one for each one of them is $f(1)\cdot k$ and the utility they gain from the agents that are in distance two (excluding agent $b$) is $f(2)\cdot k$. 
	Denote $S_1$ the set of $n$ vertices that lie between $H_1$ and $B$.
	Hence, it holds that $\sum_{j \in S_1}s_j=k$, which is a solution for the original instance of \textsc{Equitable Partition}.
\end{proof}

Our last result of the section shows that even restricting the topology to a path does not surprisingly suffice for tractability.

\begin{theorem}\label{thm:NPc_path}
	It is \NPc to decide the \IRTDG problem, even if there are only three arcs in the enmity graph, all of them are directed towards the same agent, and the topology is a path.
\end{theorem}
\begin{proof}
	Given an instance $S$ of \textsc{Equitable Partition} with $2n$ integers, we construct an equivalent instance $\mathcal{J}$ of the \IRTDG problem as follows. Recall that we assume that 
	for any $I\subseteq [2n]$ such that $|I|\le n-1$, we have $\sum_{i\in I} s_i < k$.  
	
	The topology of $\mathcal{J}$ is a path $P = (v_1, v_2, \ldots, v_{2n+4})$ on $2n+4$ vertices. The set of agents $\agents$ also contains $2n+4$ vertices, split into:
	\begin{itemize}
		\item one trouble-maker $t$;
		\item three grumpy agents $g_1$, $g_2$, $g_3$;
		\item $2n$ element-agents $a_1, a_2, \ldots, a_{2n}$.
	\end{itemize}
	
	The idea is that only negative utilities are set from the three grumpy agents towards the trouble-maker. To balance it, the grumpy agents will have positive utility towards the element agents that depend on which element the given agent represent. The most animosity is from $g_3$ towards $t$ and the function $f$ is carefully crafted, so that $g_3$ and $t$ are at the opposite sides of $P$ and $g_3$ needs $n$ element-agents, at distance at most $n$ each, representing elements with total sum at least $k$ to balance the negative contribution of $t$. The second most animosity is from $g_2$ towards $t$, crafted that $g_2$ needs to be at distance at least $n+2$ from $t$ and when it is at distance exactly $n+2$, then $g_2$ needs all the element-agents at distance at most $n$ to balance its animosity towards $t$. Finally, this will fix $g_1$ exactly next to $t$ and to balance its animosity towards $t$, we need $n$ element-agents, at distance from $g_1$ at most $n$ each, representing elements with total sum at least $k$ to balance the negative contribution of $t$.
	
	Now let $\ell$ be an integer such that $2k^3\ge 2^\ell > k^3$ and let us set $f(d)$ for $d\in [2n+4]$ as follows: 
	
	\[
	f(d) = 
	\begin{cases*}
		2^{3\ell}-d & \text{ if  $1 \le d \le n$,} \\
		2^{2\ell} & \text{ if  $d = n+1$,} \\
		2^{\ell}-d & \text{ if  $n+2 \le d \le 2n+2$,} \\
		1 & \text{ if $d= 2n+3$. }
	\end{cases*}
	\]
	
	Now for $i\in \{1,2,3\}$ and $j\in [2n]$, we let the utility $u_{g_i}(a_j) =s_j$. Recall that we assume that $s_j \ge n^2$. Observe that utility that the agent $g_i$ gets from $a_j$ at distance at most $n$ is between $(2^{3\ell}-n)s_j$  and $2^{3\ell}\cdot s_j$ and the utility obtained at distance at least $n+1$ is at most $2^{2\ell} \cdot s_j$. 
	That is if the elements for some $i\in \{1,2,3\}$, we have exactly the element-agents $\{a_i\mid i\in I\}$ for some set $I\subseteq [2n]$ at distance at most $n$ and the remaining element-agents at a larger distance, the contribution of the element-agents towards utility of $g_i$ is $2^{3\ell}\cdot \sum_{j\in I}s_j - \sum_{j\in I}(d_j\cdot s_j) + r_I$, where $d_j$ is distance between $g_i$ and $s_j$ and $r_I$ is the total utility contribution of element-agents at distance at least $n+1$, so $r_I\le 2^{2\ell}\cdot 2k \le 8k^7$. 
	Note that we are assuming $k\ge n^3$ and $n\ge 10$. So $2^{2\ell}\cdot 2k$ is much smaller than $2^{3\ell}\ge k^9$. 
	Similarly, $\sum_{j\in I}(d_j\cdot s_j)\le n\cdot 2k$ is much smaller than $2^{3\ell}$ and so the element-agents contribution to the utility of $g_i$ if 
	exactly the element-agents with indexes in $I$ are at distance at most $n$ from $g_i$ is $2^{3\ell} \cdot (\sum_{j\in I}s_j) + r'_I$, where $r'_I \ge - n\cdot 2k$ and $r'_I < 2^{3\ell}$.
	
	With this in mind, we can set the utilities of $g_1$, $g_2$, and $g_3$ towards $t$ as follows: 
	
	\begin{align*}
		u_{g_1}(t) &= - \frac{2^{3\ell}\cdot k - n\cdot 2k}{f(1)}, \\
		u_{g_2}(t) &= - \frac{2^{3\ell}\cdot 2k - n\cdot 2k}{f(n+1)}, \text{and } \\
		u_{g_3}(t) &= - \frac{2^{3\ell}\cdot k - n\cdot 2k}{f(2n+3)}.
	\end{align*}
	
	We can now verify that we indeed get the intended positioning of the agents as described above. 
	Given an assignment $\lambda$, if $\dist(\lambda(t), \lambda(g_3))\le 2n+2$, then $\dist(\lambda(t), \lambda(g_3))\cdot u_{g_3}(t) \le - (2^\ell-d) \cdot (2^{3\ell}\cdot k - n\cdot 2k) < - 2^{3\ell}\cdot 2k $ and there is no way for the element-agents to balance this. So $\dist(\lambda(t), \lambda(g_3))= 2n+3$ and $\dist(\lambda(t), \lambda(g_3))\cdot u_{g_3}(t) = 2^{3\ell}\cdot k - n\cdot 2k$. From the discussion above, we need that the subset of elements-agents that are allocated at distance at most $n$ from $g_3$ have to sum to at least $k$ and if that is the case, they can be distributed arbitrarily on vertices at distance at most $n$. Since we assume that no subset of less than $n$ elements can sum-up to $k$, it follows that at the vertices at distance $n+3, n+4, \ldots, 2n+2$ from $t$ are element-agents and the sum of elements their represent is at least $k$.
	
	Given this $g_2$ is at distance at most $n+2$ from $t$. Repeating the same argument and observing that $\frac{f(n+1)}{f(n+2)} 
	> 2^\ell$, we get that $g_2$ has to be at distance exactly $n+2$ from $t$ and that the sum of the elements represented by the element-agents at distance at most $n$ from $g_2$ is at least $2k$. That is all element-agents have to be at distance at most $n$ from $g_2$. 
	
	This means that the only possibility for a position of $g_1$ in an individually rational allocation is next to $t$. In this case $\dist(\lambda(t), \lambda(g_1))\cdot u_{g_1}(t) = -(2^{3\ell}\cdot k - n\cdot 2k)$ and we get that the subset of elements-agents that are allocated at distance at most $n$ from $g_3$ have to sum to at least $k$. Since, the sum of all elements is exactly $2k$, it follows that the sum of the elements that are associated with element agents at distance between $2$ and $n+1$ from $t$ is exactly $k$ and the sum of the elements that are associated with element-agents at distance between $n+3$ and $2n+2$ from $t$ is exactly $k$ as well. Therefore, if $\mathcal{J}$ admits individually rational solution, then $S$ admits an equitable partition. On the other hand, it is straightforward to verify that given $I\subseteq [2n]$ such that $|I|=n$ and $\sum_{i\in I}s_i=k$, we can construct an individually rational allocation $\lambda$ by letting $\lambda(t)=v_1$,
	$\lambda(g_1)=v_2$, $\lambda(g_2)=v_{n+3}$, $\lambda(g_3)=v_{2n+4}$, and placing element-agents with indexes in $I$ to vertices $v_3, \ldots, v_{n+2}$ and the remaining element-agents at vertices $v_{n+4}, \ldots, v_{n+3}$. Hence the two instances are equivalent. 
\end{proof}

In the last result of this section, we show that our problem remains intractable even if the topology is disconnected and each component is of a constant size. That is, the hardness of the problem is not caused by the fact that an agent's utility in an assignment depends on all other agents that participate in the game.

\begin{theorem}\label{thm:const_components}
	For any distance factor function $f$, it is \NPc to decide the \IRTDG problem, even if the topology is disconnected, each connected component is of size $5$, and the utilities are symmetric.
\end{theorem}
\begin{proof}
	We prove the result via a reduction from the \mbox{\textsc{3-Partition}} problem. Here, we are given a list $\mathcal{S}=(s_1,\ldots,s_{3n})$ of $3n$ positive integers such that $\sum_{i\in[3n]} s_i = n\cdot k$, and the goal is to decide whether there exists a partition of $\mathcal{S}$ to triplets $S_1,\ldots,S_n$ such that for each $S_i$, $i\in[n]$, $\sum_{s \in S_i} s = k$. \textsc{3-Partition} is known to be \NPc even if all numbers are encoded in unary and each $s_i$, $i\in[3n]$, is between $k/4$ and $k/2$~\cite{GareyJ1975,GareyJ1979}.
	
	Given an instance $\mathcal{S}$ of the \textsc{3-Partition} problem, we construct an equivalent instance $\mathcal{J}$ of the \IRTDG problem as follows. First, the topology~$G$ consists of $n$ disjoint copies of $K_5$, that is, $G$ is a disjoint union of cliques on $5$ vertices. Next, the set of agents consists of $2n$ \emph{guard-agents} $g_1,\ldots,g_{2n}$ and $3n$ \emph{element-agents} $a_1,\ldots,a_{3n}$, each corresponding to a single element of $S$. The utilities are as follows. The utility function between element-agents is constant zero, and for a guard-agent $g$, we set $\util_{a_i}(g) = s_i$ for every element-agent $a_i$, $i\in[3n]$. Let $g_i$, $i\in[2n]$, be a guard-agent. For every other guard-agent $g_j \not= g_i$, we have $\util_{g_i}(g_j) = -k$, and for a element-agents $a_j$, $j\in[3n]$, we have $\util_{g_i}(a_j) = s_j$. The distance factor function may be arbitrary, as the distances in the topology are either $1$ or $\infty$. Moreover, it is easy to see that the utilities are symmetric.
	
	The high-level idea behind the construction is that every clique contains exactly two guard-agents. Their utility, ignoring other agents, from being in the same connected components is $-k$, and the only way how to make the utility non-negative is to split the element-agents into triplets such that the increase in the utility of each guard-agent is exactly $k$. Observe that it cannot be more by the definition of the utilities.
	
	For correctness, let $\mathcal{S}$ be a \emph{yes}-instance and $S_1,\ldots,S_n$ be a desired partition of $S$. Let $C_1,\ldots,C_n$ be an arbitrary but fixed order of the connected components of $G$. We construct an individually rational assignment $\lambda$ as follows. For each $C_i$, $i\in[n]$, we assign arbitrarily to its vertices the guard-agents $g_{2i-1}$, $g_{2i}$, and all element-agents $a_j$ such that $s_j\in S_i$. For element-agents, all assignments are individually rational, and hence, we only need to check that $\lambda$ is individually rational for guard-agents. Let $C_i$, $i\in[n]$, be an arbitrary connected component of $G$, $g$ and $g'$ be two guard-agents such that $\{g,g'\} \subset V(C_i)$, and $a_{j_1}$, $a_{j_2}$, and $a_{j_3}$ be element-agents assigned to $C_i$. The utility of $g$ in assignment $\lambda$ is
	\begin{align*}
		\util_g(\lambda) &= \util_g(g') + \util_g(a_{j_1}) + \util_g(a_{j_2}) + \util_g(a_{j_3})\\
		&= \dffn(1)\cdot(-k) + \dffn(1)\cdot s_{j_1} + \dffn(1)\cdot s_{j_2} + \dffn(1)\cdot s_{j_3}\\
		&= \dffn(1)( -k + s_{j_1} + s_{j_2} + s_{j_3}) = 0.
	\end{align*}
	The equation holds because $\mathcal{S}$ is a \emph{yes}-instance, elements in each $S_i$ sum up to exactly $k$, and $\dffn(1)$ is always non-negative. The utility of guard-agent $g'$ is symmetric, and therefore, $\lambda$ is individually rational assignment witnessing that $\mathcal{J}$ is also a \emph{yes}-instance.
	
	In the opposite direction, let $\mathcal{J}$ be a \emph{yes}-instance and $\lambda$ be an individually rational assignment. First, we show that $\lambda$ assigns exactly two guard-agents to each connected component of $G$.
	
	\begin{claim}
		For each connected component $C$ of $G$, it holds that $|\{\lambda(g_1),\ldots,\lambda(g_{2n})\}\cap C| = 2$.
	\end{claim}
	\begin{claimproof}
		For the sake of contradiction, suppose that there exists a component $C$ such that $\lambda$ assigns at least three guard-agents $g_{i_1}$, $g_{i_2}$, and $g_{i_3}$ to its vertices. Then, the utility of $g_{i_1}$ is
		\begin{align*}
			\util_{g_{i_1}}(\lambda) &\leq \util_{g_{i_1}}(g_{i_2}) + \util_{g_{i_1}}(g_{i_3}) + \dffn(1)\cdot\frac{k}{2} + \dffn(1)\cdot\frac{k}{2} \\
			&= \dffn(1)\cdot (-k) + \dffn(1)\cdot (-k) + 2\cdot\dffn(1)\cdot\frac{k}{2}\\
			&= \dffn(1)( -2k + k ) = \dffn(1)(-k).
		\end{align*}
		Again, as $\dffn(1)$ is always positive, we obtain the $\lambda$ is not individually rational. Hence, each individually rational assigns at most two guard-agents to each connected component. By the Pigeonhole principle, we obtain that the number of assigned guard-agents is exactly two.
	\end{claimproof}
	
	Now, we create a solution partition $S_1,\ldots,S_n$ for $\mathcal{S}$ such that for each $S_i$, we set $S_i = \{s_j\mid \lambda(a_j) \in V(C_i)\}$. For the sake of contradiction, suppose that there exists a set $S_i$, $i\in[n]$, such that $\sum_{s\in S_i} s \not= k$. This means that for guard agent $g$ such that $\lambda(g) \in V(C_i)$, the utility was $\dffn(1)\cdot(-k) + (k - \varepsilon) = - \varepsilon$, where $\varepsilon > 0$. However, this would mean that $\lambda$ was not individually rational, which is a contradiction. Therefore, such a situation cannot occur. Moreover, it is easy to see that $S_1,\ldots,S_n$ is indeed a partition of $\mathcal{S}$, and the theorem follows.
\end{proof}

\section{Parameter-Many Agents}\label{sec:agents_parameter}

In the previous section, we have established strong intractability results for the problem when the number of agents is part of the input. 
For this reason, in this section, we follow the parameterized complexity paradigm and we consider $\numagents$ to be a parameter of the problem; \Cref{fig:results:agentsParam} provides a mind-map of our results. Note that parameterization by the number of agents has been successfully used to give fixed-parameter algorithms for various hard problems in computational social choice; see, e.g.,~\cite{BredereckFKKN2020,DeligkasEGHO2021,GanianOR2023}.

\begin{figure}[tb!]
	\centering
	\begin{tikzpicture}[align=center,node distance=1.8cm]
		\tikzstyle{decision} = [draw,align=center,text width=2.5cm];
		\tikzstyle{result} = [draw,rectangle,rounded corners,align=center,text width=2.5cm,node distance=4cm];
		\tikzstyle{Wh} = [fill=orange!30];
		\tikzstyle{NPh} = [fill=red!30];
		\tikzstyle{FPT} = [fill=green!30];
		
		\node[decision] (d1) at (0,0) {Utilities};
		\node[decision,below of=d1] (d2) {Topology};
		\node[decision,below of=d2] (d3) {Enmity Graph};
		\node[decision,below of=d3] (d4) {Distance Factor Fun.};
		
		\node[right of=d1,node distance=4cm] (dummy) {};
		\node[result,Wh,right of=d1,node distance=7cm] (p1) {\small\Wc\\\small[Thm.~\ref{thm:IRTDG:Wh:agents}]};
		\node[result,Wh,below of=dummy,node distance=0.8cm] (p2) {\small \Wc\\\small[Thm.~\ref{thm:Wh:path}]};
		\node[result,FPT,below of=p2,node distance=1.2cm] (p3) {\small\FPT\\\small[Cor.~\ref{thm:FPT_sd}]};
		
		\node[decision,right of=d3,node distance=4cm] (d5) {Topology};
		\node[result,FPT,right of=d5,node distance=4cm,yshift=0.75cm] (p6) {\small\FPT\\\small[Thm.~\ref{thm:FPT_path}]};
		\node[result,Wh,below of=p6,node distance=1.2cm] (p7) {\small\Wc\\\small[Thm.~\ref{thm:Wc:two_types}]};
		\draw[->] (d5) -- (p6) node [above,midway,sloped] {\it\small path};
		\draw[->] (d5) -- (p7) node [below,midway,sloped] {\it\small any};
		
		\node[result,FPT,right of=d4,node distance=7cm] (p4) {\small\FPT\\\small[Thm.~\ref{thm:FPT_cw}]};
		
		\draw[->] (d1) -- (p1) node [above,midway] {\it\small symmetric};
		\draw[->] (d1) -- (d2) node [left,midway] {\it\small any};
		\draw[->] (d2) -- (p2) node [above,midway,sloped] {\it\small path};
		\draw[->] (d2) -- (p3) node [below,midway,sloped] {\it\small \sd};
		\draw[->] (d2) -- (d3) node [left,midway] {\it\small any};
		\draw[->] (d3) -- (d5) node [above,midway,sloped] {\it\small in-star};
		\draw[->] (d3) -- (d4) node [left,midway] {\it\small any};
		\draw[->] (d4) -- (p4) node [above,midway] {\it\small bounded + \tww};
	\end{tikzpicture}
	\caption{A simplified overview of our results for the setting with parameter-many agents. All \Wc combinations can be solved by an \XP algorithm, which is asymptotically optimal under ETH (see \Cref{thm:XP_agents}). We use \tww{} to highlight that the result additionally requires the topology of bounded twin-width, and similarly \sd{} represents topologies with bounded shrub-depth.}
	\label{fig:results:agentsParam}
\end{figure}
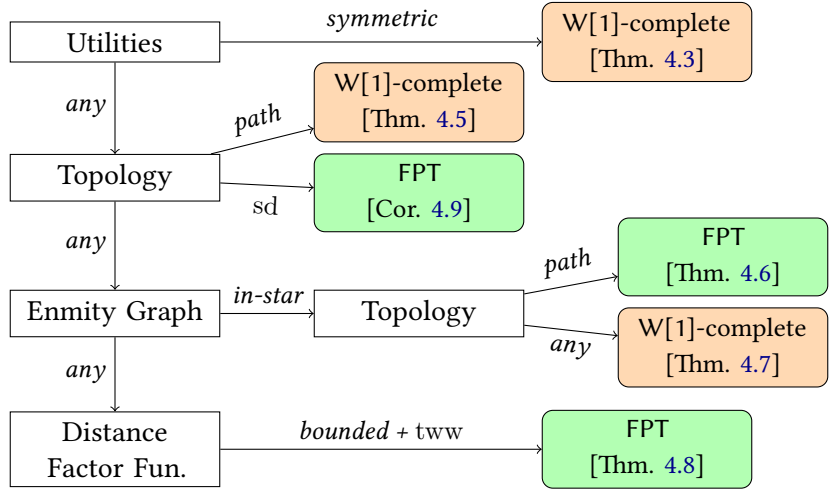

Our first result is a brute-force algorithm that finds an individually rational assignment (if one exists) in \XP time. In other words, the \IRTDGshort problem is solvable in polynomial time if the number of agents is a fixed constant.

\begin{theorem}\label{thm:XP_agents}
	There is an algorithm for the \IRTDG problem running in $|V(G)|^\Oh{\numagents}$ time.
\end{theorem}
\begin{proof}
	The algorithm is a simple brute-force. We exhaustively try all assignments of vertices of the topology to agents. Then, in polynomial time, we verify whether the checked possibility assigns to each agent a different vertex and whether the assignment is individually rational. If this is the case, we return \emph{yes} as the result. Otherwise, if no possibility leads to an individually rational assignment, we return \emph{no}. The algorithm is trivially correct as it checks all possible assignments. Additionally, there are $V(G)^\Oh{\numagents}$ possible agents-vertices assignments, and for each assignment, the verification of uniqueness of the vertices and of the individual rationality can be performed in polynomial-time. Therefore, the overall running time is $|V(G)|^\Oh{\numagents}$.
\end{proof}

\begin{lemma}\label{lem:IRTDG:W:agents}
	For every distance factor function $\dffn$, the \IRTDG problem parameterized by the number of agents $|\agents|$ is in \W[1].
\end{lemma}
\begin{proof}[Proof (sketch)]
	To show that \IRTDGshort belongs to \W[1], we provide a parameterized reduction to the \textsc{Short Turing Machine Acceptance} problem. Given an instance $\mathcal{I}=(G,\agents,\util,\dffn)$, we first compute all-pairs shortest path for all vertices of~$G$ using a known algorithm~\cite{Floyd1962,Warshall1962}. We use this information for the construction of the transition function of the equivalent Turing machine. The basic idea is that the Turing machine first guesses the positions of all agents. For this, the alphabet of the Turing machine contains one symbol for each vertex of~$G$. Then, the machine verifies whether every agent is assigned to a different vertex. Next, for every agent $i\in\agents$, it computes the utility $\util_i$ by enumerating the guessed positions of all other agents and checks whether this value is non-negative. This can be clearly done in $f(k)$-many steps, as the distances and the increase of utility can be encoded in the transition function using the computed distances and the distance factor function.
\end{proof}

Now, the natural question arises. Is the \XP algorithm of \Cref{thm:XP_agents} the best we can hope for, or is there an \FPT algorithm for the problem? We resolve this question in negative in our next result.

\begin{theorem}\label{thm:IRTDG:Wh:agents}
	For every distance factor function $\dffn$, it is \Wc to decide the \IRTDG problem parameterized by the number of agents $|\agents|$, even if the utilities are symmetric, the utility function uses two different values, and there are only two types of agents.
\end{theorem}
\begin{proof}
	We provide a parameterized reduction from the \textsc{Independent Set} problem, which is very well-known to be \Wc when parameterized by the solution size~$k$~\cite{DowneyF1995}. Let $\mathcal{I} = (H,k)$ be an instance of the \textsc{Independent Set} problem. We construct an equivalent instance $\mathcal{J}$ of the \IRTDG problem as follows.
	
	First, the topology $G$ is just a copy of the graph $H$ with one added apex vertex~$x$. The set of agents consists of $k$ \emph{standard} agents $a_1,\ldots,a_k$ and a single \emph{guard} agent $g$. Next, let $\beta \in \mathbb{R}_{>0}$ be a number. 
	We fix an arbitrary distance factor function~$\dffn$. Finally, we define the utilities. For every pair of distinct standard agents $a_i,a_j\in\agents$, we set $\util_{a_i}(a_j) = \util_{a_j}(a_i) = -\beta$ and $\util_{a_i}(g) = \util_g(a_j) = \frac{(k-1)\cdot\dffn(2)\cdot\beta}{\dffn(1)}$. Clearly, there are two types of agents, and the utilities are symmetric.
	
	For the correctness, let $\mathcal{I}$ be a \emph{yes}-instance and $S\subseteq V(H)$ be an independent set of size~$k$. We assign the standard agents to vertices of~$G$ corresponding to vertices in~$S$, and we assign the guard to the apex vertex~$x$. Let $\assgn$ be the described assignment. Since~$S$ is an independent set, the distance between all pairs of standard agents is at least two. Additionally, the guard agent is a direct neighbor of all standard agents. Therefore, the utility of each standard agent $a_i$ is 
	\begin{multline*}
		\left(\sum_{j\in[k]\setminus\{i\}} \dffn(\dist(\assgn(a_i),\assgn(a_j)))\cdot \util_{a_i}(a_j)\right) + \dffn(1)\cdot\util_{a_i}(g) 
		= \left(\sum_{j\in[k]\setminus\{i\}} \dffn(2)\cdot (-\beta) \right) + \dffn(1)\cdot\util_{a_i}(g) \\= -((k-1)\cdot\dffn(2)\cdot\beta) + \dffn(1)\cdot\frac{(k-1)\cdot\dffn(2)\cdot\beta}{\dffn(1)} = -((k-1)\cdot\dffn(2)\cdot\beta) + (k-1)\cdot\dffn(2)\cdot\beta = 0.
	\end{multline*}
	Hence, as $g$ gains positive utility from all standard agents, $\assgn$ is individually rational.
	
	In the opposite direction, let $\mathcal{J}$ be a \emph{yes}-instance and let~$\assgn$ be an individually rational assignment. First, we show that no pair of standard agents are neighbors with respect to $\assgn$.
	
	\begin{claim}
		Let $\assgn$ be an individually rational solution. Then, for all pairs of distinct standard agents $a_i$ and $a_j$, it holds that $\{\assgn(a_i),\assgn(a_j)\}\not\in E(G)$.
	\end{claim}
	\begin{claimproof}
		For the sake of contradiction, suppose that it is the case and let $a_i$ and $a_j$ be standard agents occupying neighboring vertices. The utility of the agent $a_i$ is then
		\begin{multline*}
			\left(\sum_{\ell\in[k]\setminus\{i\}} \dffn(\dist(\assgn(a_i),\assgn(a_\ell)))\cdot \util_{a_i}(a_\ell)\right)
			+ \dffn(\dist(\assgn(a_i),\assgn(g)))\cdot \util_{a_i}(g)\\
			\leq -\dffn(1)\cdot \beta - (k-2)\cdot\dffn(2)\cdot\beta + \dffn(1)\cdot\util_{a_i}(g)\\
			= -\dffn(1)\cdot \beta - (k-2)\cdot\dffn(2)\cdot\beta + \dffn(1)\cdot\frac{(k-1)\cdot\dffn(2)\cdot\beta}{\dffn(1)}\\
			= -\dffn(1)\cdot \beta - (k-2)\cdot\dffn(2)\cdot\beta + (k-1)\cdot\dffn(2)\cdot\beta
			= -\dffn(1)\cdot \beta + \dffn(2)\cdot\beta,
		\end{multline*}
		from which it clearly follows that the utility of $a_i$ is strictly negative, as $\dffn(1) > \dffn(2)$. Therefore, such an assignment $\assgn$ would not be individually rational.    
	\end{claimproof}
	
	Consequently, for every individually rational assignment $\assgn$ we have that no pair of standard agents is assigned to neighboring vertices. Thus, setting $S = \{ v \mid \exists i \in [k]\colon \assgn(a_i) = v \}$ we create a solution for $\mathcal{I}$, and the instances are equivalent.
	
	Clearly, the construction can be done in polynomial time. Moreover, we have $\numagents = k + 1$, and therefore, the reduction is clearly a parameterized reduction, finishing the \Whness part of the proof, which, together with \Cref{lem:IRTDG:W:agents}, implies \Wc{}ness of the problem.
\end{proof}

Consequently, if we parameterize only with the number of agents, an \FPT algorithm cannot exist (unless $\FPT=\W[1]$). What is even more disturbing is that the simple brute-force algorithm proposed in \Cref{thm:XP_agents} is, under standard theoretical assumptions, asymptotically optimal.

\begin{theorem}
	Unless ETH fails, there is no algorithm solving the \IRTDG problem in $g(\numagents)\cdot |V(G)|^{o(\numagents)}$-time for any computable function~$g$.
\end{theorem}
\begin{proof}
	Recall the construction from \Cref{thm:IRTDG:Wh:agents}.  It is well-known that, unless ETH fails, there is no algorithm solving the \textsc{Independent Set} problem in $g(k)\cdot n^{o(k)}$ time for any computable function $g$~\cite{ChenCFHJKX05}. For the sake of contradiction, assume that there exists an algorithm $\mathbb{A}$ that solves \IRTDG in $h(\numagents)\cdot |V(G)|^{o(\numagents)}$ time for some computable function $h$. Then, assume the following algorithm for \textsc{Independent Set}. Given an instance~$\mathcal{I}$, use the construction above to produce an equivalent instance of the \IRTDG problem, call the algorithm $\mathbb{A}$, and return the same result for $\mathcal{I}$. The construction of the equivalent instance takes polynomial time and the algorithm runs in $h(\numagents)\cdot |V(G)|^{o(\numagents)} = h(k+1)\cdot |V(G)|^{o(k+1)}$ time. Overall, we obtain an algorithm for \textsc{Independent Set} running in $g(k)\cdot n^{o(k)}$ time, which is unlikely.
\end{proof}

The previous results clearly indicate that, in order to reveal at least some tractability, we need to further restrict the input instances. 
We start with a very strong intractability result, which shows that when the distance factor function remains unrestricted, there cannot be an \FPT algorithm with respect to the number of agents, even under the severe restriction of having a path topology.
The proof is based on a reduction from (a special case of) the following problem.

\prob{\textsc{Partitioned Subgraph Isomorphism} (PSI)}
{Two undirected graphs $G$ and $H$ with $|V(H)| \le |V(G)|$ ($H$ is \emph{smaller}) and a mapping $\psi\colon V(G) \to V(H)$.}
{Is $H$ isomorphic to a subgraph of $G$? I.e., is there an injective mapping $\phi\colon V(H) \to V(G)$ such that $\{\phi(u),\phi(v)\} \in E(G)$ for each $\{u,v\} \in E(H)$ and $\psi \circ\phi$ is the identity?}

\begin{theorem}\label{thm:Wh:path}
	\IRTDG is \Wc parameterized by the number of agents, even if the topology is a path. Unless ETH fails, there is no 
	algorithm solving the \IRTDG problem in $g(\numagents)\cdot |V(G)|^{o(\frac{\numagents}{\log\numagents})}$-time 
	for any computable function~$g$.
\end{theorem}
\begin{proof}
	We show \Whness by a parameterized reduction from the \textsc{Partitioned Subgraph Isomorphism} problem, which is known to be \Wc when parameterized by the solution size~$k$ even on $3$-regular graphs. Furthermore, there is no algorithm $\mathbb{A}$ and function $g$ such that $\mathbb{A}$ correctly decides every instance of PSI with the smaller graph $H$ being 3-regular in time $g(|V(H)|)n^{o(|V(H)|/\log |V(H)|)}$, unless ETH fails (see {\cite{Marx2010}} and \cite{EibenKPS19}).
	
	Let $(G,H, \psi)$ be an instance of PSI with $H$ 3-regular and denote \mbox{$k=|V(H)|$}. 
	Note that the mapping $\psi\colon V(G) \to V(H)$ partitions the vertices of $V(G)$ into $x=|V(H)|$ many parts $V_1,\ldots, V_x$, each corresponding to a specific vertex of~$H$. Moreover, we wish to select in each part $V_i$, $i\in [x]$, exactly one vertex $v_i$, such that if $vw\in E(H)$ and $V_i$ corresponds to $v$ and $V_j$ corresponds to $w$, then $v_iv_j$ is an edge in $G$. Notice that if $vw\notin E(H)$, then the edge in $v_iv_j$ is not required to be in $E(G)$, however, it is also not forbidden. Hence we can, without loss of generality, assume that $G$ does not contain edges between $V_i$ and $V_j$ if these two vertex sets correspond to vertices in $H$ that are not adjacent. It follows that we can also partition the edges of $E(G)$ into $y= |E(H)|$ many parts $E_1, \ldots, E_y$, each corresponding to a specific edge of $H$. This is important, because, as is usual for a reduction from PSI, we will have a ``gadget'' to select a single vertex in each $V_i$, a gadget to select an edge in each $E_j$, and then a way to check that this selection is consistent.
	
	We will now construct an equivalent instance $\mathcal{J}$ of the \IRTDG problem such that the topology of $\mathcal{J}$ is a path on $|V(G)|^\Oh{1}$ many vertices. A very crude idea of this reduction is to assign each of the vertex-parts $V_i$, $i\in [x]$,  and each of the edge-parts $E_j$,  $j\in [y]$ an interval on the path such that these intervals are disjoint. Then using two additional ``guard'' vertices placed at the endpoints of the path, $\Oh{|V(H)|+|E(H)|}$ many so-called ``anchor'' vertices, and a clever choice of the function $f$, we force each of the intervals to contain exactly two consecutive vertices at some ``allowed'' positions inside the interval, that represent a selection of specific edge or vertex in this interval. Finally, using basically the same trick we used to force the ``allowed'' consecutive vertices in a selection gadget to be only at specific positions -- specific distances from an ``anchor'' vertices -- we are able to force that the selected vertices $v_i\in V_i$ and the selected edges $e_j\in E_j$ are consisted, i.e., if $E_j$ is associated with an edge of $H$ whose one endpoint is the vertex associated with $V_i$, then $v_i$ is an endpoint of $e_j$. It is important that the intervals for these selection gadgets are placed carefully and all the distances for which we need to set up the value of $f$ are different.
	
	Now let us go a bit more into detail. For ease of notation let $n=|V(G)|$, $m=|E(G)|$. For the sake of exposition, we also assume that $n$ is large enough, say $n > 100$. 
	We can use this assumption without loss of generality, as PSI on instances with $n \le 100$ is solvable in constant time by trying all possibilities. Note that this assumption is not strictly necessary for the reduction, but allows us a slightly ``cleaner'' definition of the function $f$ without additional constants that would be necessary to handle cases when $n$ is small.  
	We let the topology of $\mathcal{J}$ be the path $P$ on $40mn$ vertices.  
	
	The set of agents $\agents$ consists of the following.
	\begin{itemize}
		\item Two guard agents $g_1,g_2$, these will be placed at the endpoints of $P$ and the intervals for ``selector'' gadgets will be defined by their distance from $g_1$.
		\item A "dummy" agent $d_2$ to make $g_2$ ``happy'' if $g_1$ is at the other endpoint and $d_2$ exactly next to it.
		\item $x = |V(H)|$ many ``vertex-selector'' pairs of agents $v_i, w_i$, $i\in [|V(H)|]$. The idea is that $v_i$ represents the selection of the vertex in $V_i$ and $w_i$ is a helper agent that fixes $v_i$ in the interval for $V_i$.
		\item $y = |E(H)|$ many ``edge-selector'' pairs of agents $e_j, e'_j$, $j\in [|E(H)|]$. Again $e_j$ represents the selection of an edge in $E_j$ and $e'_j$ is the helper agent to fix $e_j$ in the interval of $P$ selected for $E_j$.
		\item $x+y$ many ``anchor'' pairs of agents  $a_i, b_i$ $i\in [x+y]$. These are designed such that $a_i$ is at the start of each ``selector'' gadget, and we use them together with the distance factor function $f$ to force the selector-pair to occupy only a specified subset of vertices inside the selector gadget.  The agent $b_i$ is again a ``helper'' agent that will be placed next to $a_i$.
	\end{itemize}

	We will now describe how we set the distance factor function $f$ and how the utilities of the agents are defined. We remark that if we do not define a utility $u_i(j)$ of an agent $i$ toward agents $j$, then we assume that $u_i(j) = 0$. The function $f$ is such that for any $d\in \{1,\ldots, |V(P)|-1\}$ we have $f(d) = 2^{p_d(n)}-q_d$, where $0\le q_d < d$ and $n^3 \ge p_d(n)\ge n^3 - 7n^2$.
	In addition, steps from $d$ to $d+1$ in the function $f$ are always one of the following four types: (i) $f(d+1) = f(d)-1$, (ii) $f(d+1) = 2^{p_d(n)-1}$, (iii)
	$f(d+1) = 2^{p_d(n)-1}-q_d$, and (iv) $f(d+1) = 2^{p_d(n) -  n}$. 
	We start with $f(1) = 2^{n^3}$ and $f(2) = 2^{ n^3-n}$. Unless we specify otherwise, we have $f(d+1) = f(d)-1$.  
	
	Let us first fix $g_1$ and $g_2$ to be at the endpoints of $P$. This is easily done by setting $u_{g_2}(d_2) = 1$, $u_{g_2}(g_1) = - \frac{f(1)}{f(|V(P)|- 1)}$ and letting utility of $g_2$ towards any other agent to be $0$. 
	Since $f$ is a decreasing function, it follows that for any assignment $\lambda$, we have 
	
	\begin{flalign*}
		u_{g_2}(\lambda) &= f(\dist_P(\lambda(g_2), \lambda(d_2))) 
		\quad	- \frac{f(1)}{f(|V(P)|- 1)} \cdot f(\dist_P(\lambda(g_2), \lambda(g_1))) && \\
		&\le   f(1) - 
		\frac{f(1)}{f(|V(P)|- 1)} \cdot f(\dist_P(\lambda(g_2), \lambda(g_1))) && \\
		&\le   f(1) - f(1) = 0 &&
	\end{flalign*}
	where inequality is achieved if and only if  $\dist_P(\lambda(g_2), \lambda(d_2)) = 1$ and $\dist_P(\lambda(g_2), \lambda(g_1)) = |V(P)| - 1$. Therefore, for an assignment $\lambda$ to be individually rational, we necessarily have $g_1$ and $g_2$ are at the endpoints of $P$. 
	
	For the sake of the exposition, we assume that $g_1$ is on the left endpoint of $P$, i.e., whenever we say ``left'' we mean closer to $g_1$ and by ``right'' we mean closer to~$g_2$. 
	
	Note that all the remaining agents come in pairs, where one agent is ``selector'' or ``anchor'' and the other is ``helper''. The idea is that we use the utility of the ``helper'' agent towards $g_1$, $g_2$, and its pair to force its pair next to the helper and with distance to $g_1$ being some fixed interval of length less than $2mn$.  Moreover, this interval will be in distance at least $10mn$ and at most $17mn$ from $g_1$ (so in distance at least $23mn-1$ from $g_2$).
	Let $(s,h)$ be a pair of agents -- $s$ a selector/anchor agent and $h$ a ``helper'' agent. And let $\ell,r \in \mathbb{N}$ such that $\ell<r<\ell + 2mn$. Assume that we want to force that $h$ to be placed in the distance at least $\ell+1$ and at most $r-1$ from $g_1$ and at the same time $s$ to be next to $h$. In order to do this, we will need to fix values of $f(i)$ for $i\in \{\ell, \ell+1, \ldots, r\}$ and for $i\in \{|V(P)|-1-r, |V(P)|-r, \ldots, |V(P)|-1-\ell\}$. 
	We set it as follows: 
	
	\begin{itemize}
		\item $f(\ell) = 2^{p_\ell(n)}$ for some polynomials $p_\ell(n) = n^3 - \Oh{n^2}$ such that $\frac{f(\ell-1)}{2} \le f(\ell) < f(\ell-1)$ (i.e., there is type-(ii) step from $f(\ell-1)$ to $f(\ell)$);
		\item $f(\ell+1) = 2^{p_a(n)- n}$ (type-(iv) step);
		\item $f(|V(P)|-1-r) = 2^{p_r(n)}$ for some polynomials $p_r(n) = n^3 - \Oh{n^2}$ such that $\frac{f(|V(P)|-2-r)}{2} \le f(|V(P)|-1-r) < f(|V(P)|-2-r)$;
		\item $f(|V(P)|-r) = 2^{p_a(n)- n}$ (type-(iv) step);
		\item $f(i+1) = f(i)-1$ for $i\in \{\ell+1, \ldots, r\}\cup \{|V(P)|-r, \ldots, |V(P)|-\ell \}$ (type-(i) step);
	\end{itemize}
	
	The goal is now to set up the utilities of $h$ towards $g_1$, $g_2$, and $s$ so that the negative contribution of $g_1$ overcomes the positive contribution of $s$ if $h$ is closer to $g_1$ then at distance $\ell+1$, the negative contribution of $g_2$ overcomes the contribution of $s$ if $h$ is closer to $g_2$ than $|V(P)|-r$. In addition, we utilize the fact that all steps within this range of distances from $h$ to $g_1$ and $g_2$ are type-(i) and that the jump between $f(1)$ and $f(2)$ is type-(iv) to achieve that $s$ has to be next to $h$ else the negative contribution of $g_1$ and $g_2$ overcomes the positive contribution of $s$.
	
	We set, %
	
	\begin{itemize}        
		\item $u_{h}(s) = 1$; 
		\item $u_{h}(g_1) = - \frac{f(1)}{4f(\ell+1)} = 2^{n^3-p_\ell(n)+ n-2}$;
		\item $u_{h}(g_2) = - \frac{f(1)}{4f(|V(P)|-r)} = 2^{n^3-p_r(n)+ n -2}$.
	\end{itemize}
	It is easy to verify that in order for the total contribution of $g_1$ and $g_2$ towards the utility of $s$ to be at least $-f(1)$, the distance between $g_1$ and $h$ has to be at least $\ell+1$ and at most $r-1$. In this case, the contribution of each is at most $-\frac{f(1)}{4}$. Recall that distance between $\ell$ and $r$ is at most $2mn\le 2n^3$. From the construction of $f$ it follows that in this case, the contribution of $g_1$ is at most $-(f(\ell+1) - 2n^3)\cdot \frac{f(1)}{4f(\ell+1)} = -f(1)\cdot (\frac{1}{4} - \frac{2n^3}{f(\ell+1)}) > -f(1)\cdot (\frac{1}{4} - \frac{2n^3}{2^{n^3}- 7n^2})$. By the choice of $f$ and assumption that $n > 100$, we get that the contribution of $g_1$ is between $-\frac{f(1)}{4}$ and   $-\frac{f(1)}{8}$. Analogously, we get that contribution of $g_2$ towards the utility of $h$ is between $-\frac{f(1)}{4}$ and   $-\frac{f(1)}{8}$. Since $f(2) < \frac{f(1)}{4}$, we get that in order for $h$ to get non-negative total utility, $s$ has to be next to $h$.

	If the $(s,h)$ is an anchor-pair then we are fixing the helper $\beta$ agent at a single possible distance. That is $r=\ell+2$ above, and we are fixing $h$ at distance $\ell+1$ from $g_1$. In every anchor-pair, we want to use the utility of $g_1$ to fix $s$ to be to the left of $h$. We do it by setting $u_{g_1}(s) = \frac{f(1)}{f(\ell)}$ and $u_{g_1}(h) = - \frac{f(1)}{f(\ell+1)}$. Since we already know that $h$ has to be at distance $\ell+1$ from $g_1$, we get that if $s$ is at distance $\ell$, then the contribution of this pair towards the utility of $g_1$ is $0$ and otherwise it is negative. In order for all the contributions to be $0$, in all anchor pairs, we get that the anchor-agent $s$ is always to the left of the helper-agent $h$.
	
	Let us now describe the edge-selector gadget for the edge-set $E_j$, $j\in [y]$. 
	In the edge selection gadget for $E_j$, we will have the anchor-pair $(a_{x+j}, b_{x+j})$ fixed such that $a_{x+j}$ is at some distance $\ell_{x+j}$ from $g_1$, $b_{x+j}$ is at distance $\ell_{x+j}+1$; $\ell_{x+j}$ will be between $15mn$ and $16mn$ as seen later. We restrict the edge-selector pair $(e_j,e'_j)$ such that $e'_j$ is at distance between $\ell_{x+j}+5$ and $\ell_{x+j}+4+|E_j|$ (so there are $|E_j|$ possible positions of $e'_j$).  
	Now to make sure that $e_j$ is to the left of $e'_j$, we let $u_{a_{x+j}}(e_j) = 1$ and $u_{a_{x+j}}(e_j) = -1$. Hence, $e_j$ is only allowed to be at distance from $g_1$ that is between $\ell_{x+j}+4$ and $\ell_{x+j}+3+|E_j|$. Each such possible position of $e_j$ corresponds to a selection of a single edge in $E_j$). 
	
	The next anchor for the edge-set $E_{j+1}$ is always placed at the position $\ell_{x+j}+|E_j| + 10$. This way all the agents that are forced to be placed in the edge-selector gadget, will be at distance at least some $\ell_{x+1}$ and at most $\ell_{x+1} + m+10y < \ell_{x+1} + 2m$, where $\ell_{x+1} = 15mn + \Oh{x}$. 
	We will use the fact that all edge-gadgets are inside a rather compact region of length less than $2m$ in designing the vertex-selector gadget. Any two possible positions for a vertex-selector pair of agents will be at distance $2m$, so each possible placement of a vertex selector will have unique distances toward the possible placements of edge-selector pairs and we will be able to set the value of $f$ for each of these distances separately.
	
	For the vertex-selector gadget associated with the vertex set $V_i$, we will again have anchor-pair $(a_i, b_i)$ fixed at distance $\ell_i$, $\ell_i+1$, respectively, where $\ell_i\in \{10mn,\ldots, 12mn+ 10x\}$ and $\ell_1 = 10mn$, the exact positions of $\ell_i$ depends on $|V_1|, |V_2|, \ldots, |V_{i-1}|$. In the pair $(v_i, w_i)$ we allow $w_i$ to be at distance between  $\ell_i+5$ and $\ell_i + 5 + 2m\cdot (|V_i| -1) $ from $g_1$ and the distance for the next anchor pair $\ell_{i+1} = \ell_i + 2m\cdot |V_i|$. 
	Now we let $u_{a_i}(v_i) = 1$ and $u_{a_i}(w_i) = -2$. We note that each of the distances from $g_1$ and $g_2$ for which we described how to set the values of $f$ above will be more than $10mn$, so we are free to now set $f(d)$ for $d\in \{3, \ldots, 2mn+5\}$. Starting from $f(2) = 2^{n^3- n}$ (as set before), we let for $d\in \{2, \ldots, 2m\cdot n+4\}$, $f(d) = 2^{p_d(n)} - q_d$ for some polynomial $p_d$ and constant $q_d$ depending on $d$ (note for $d= 2$ we have $p_2(n)= n^3- n$ and $q_2(n) = 0$) and we define: 
	
	\[
	f(d+1) = 
	\begin{cases*}
		2^{p_d(n)-1}-q_d & \text{ if  $d \equiv 4 \mod 2m$} \\
		2^{p_d(n)}-q_d - 1 & \text{ otherwise }
	\end{cases*}
	\]
	
	Note that $q_d\le d$. The utility of the anchor agent $a_i$ if $v_i$ is at distance $d$ from $a_i$ and $w_i$ at distance $d+1$ is $f(d) - 2f(d+1)$. Since $d$ has to be between $4$ and $2m\cdot (|V_i|-1) + 4 < 2m\cdot n+4$, to keep the utility of $w_i$ non-negative 
	\[
	f(d) - 2f(d+1) = 
	\begin{cases*}
		q_d \quad \text{ if } d \equiv 4 \mod 2m \\
		-2^{p_d(n)}+q_d + 1 \quad  \text{ otherwise }
	\end{cases*}
	\]
	
	Hence, in order for the utility of the anchor $a_i$ to be positive, $v_i$ has to be at distance from $g_1$ that is of form $\ell_i+4+ \iota \cdot 2m$ for $\iota\in \{0, \ldots, |V_i| - 1\}$, each distance representing a single and distinct vertex in $V_i$. 
	
	Finally, the position of the vertex gadgets is to the left of the edge gadgets and we leave a gap of $3 n\cdot m$ between the rightmost allowed position of any vertex-selector pair and the position of the leftmost anchor pair for any edge-selector gadget (so $\ell_{x+1}$ is $3 n\cdot m$ many positions to the right of the last possible allowed position of $v_x$ and we get $\ell_{x+1} = 15m(n-1) + \Oh{x}$). This way the function $f$ is not yet set for any distance from the closest distance between any two positions from which one is in a vertex-gadget and the other in an edge-gadget. 
	Recall that each allowed position for $e_j$ in an edge-gadget represents a single edge in $E_j$ and each allowed position for $v_i$ in a vertex-gadget represents a single vertex in $V_i$. For $v\in V(G)$ and $e\in E(G)$,
	let $d_v^e$ be the distance between the position in a vertex-gadget representing the vertex $v$ and the position in an edge-gadget representing the edge $e$. Note that because of the gaps of length at least 
	$2m$ between any two allowed positions for a vertex-selector pair and gaps of length strictly less than $2m$ between any two allowed positions for an edge-selector pair, we have that if $d_v^e = d_w^f$ for $v,w\in V(G)$ and $e,f\in E(G)$, then $v=w$ and $e=f$. 
	Now let $D =\{d_v^e\mid v\in e\}$ be the set of distances between a position representing a vertex $v\in V(G)$ and a position representing an edge $e\in E(G)$ such that $v$ is an endpoint of $e$. Note that for all $d\in D$ we have $3m\cdot n\le d \le 2m\cdot n + 3m\cdot n + 2m \le  6m\cdot n$.
	
	Note that for $d= 2m\cdot n+5 $, we have already set $f(d) = 2^{p_d(n)}-q_d$, where $p_d(n) = n^3-\Oh{n^2}$ is a polynomial. Now for $d\in \{2m\cdot n+5, \ldots, 6m\cdot n\}$, where $f(d) = 2^{p_d(n)}-q_d$ we let: 
	
	\[
	f(d+1) = 
	\begin{cases*}
		2^{p_d(n)-1}-q_d & \text{ if  $d\in D$} \\
		2^{p_d(n)}-q_d - 1 & \text{ otherwise }
	\end{cases*}
	\]
	
	Finally, if $V_i$ is associated with a vertex $v\in V(H)$ and $E_j$ is associated with an edge $e\in E(H)$ such that $v$ is an endpoint of $V_i$, then we set the utility of the agent $v_i$ towards agents $e_j$ and $e'_j$ such that $u_{v_i}(e_j)= 1$ and $u_{v_i}(e'_j)= -2$. It follows that the contribution of $e_j$ and $e'_j$ in this case towards the utility of $v_i$ (assuming that their positions are consistent with restrictions described above to keep $g_1, g_2$, anchor-pairs, $w_i$ and $e'_i$ individually rational) is 
	
	\[
	f(d) - 2f(d+1) = 
	\begin{cases*}
		q_d \quad \text{ if } d\in D \\
		-2^{p_d(n)}+q_d + 1 \quad  \text{ otherwise }
	\end{cases*}
	\]
	
	However, $d\in D$ if and only if $v_i$ and $e_j$ represent the selection of a vertex $v\in V(G)$ and an edge $e\in E(G)$, respectively, such that $v$ is an endpoint of $e$. Furthermore, note that the degree of every vertex in $H$ is exactly three, so $v_i$ has utility for three edge-selector agent pairs. Let $V_i$ represent a vertex in $H$ that is endpoint of edges represented by edge-sets $E_{j_1}$, $E_{j_2}$, and $E_{j_3}$ and let $d_1 < d_2 < d_3$ be the distances from assigned position of $v_i$ to assigned positions of $e_{j_1}$, $e_{j_2}$,  and $e_{j_3}$, respectively.
	Notice that $q_d$ is always at most $d$ and hence we have $2^{p_{d_1}(n)} > 2^{p_{d_2}(n)} > 2^{p_{d_3}(n)} > d_1+d_2+d_3+1 \ge q_{d_1} + q_{d_2} + q_{d_3} + 1$. Hence, the only way that $v_i$ receives a non-negative utility is if all of the distances $d_1$, $d_2$, and $d_3$ are in $D$. However, this means that the vertex in $G$ associated with the position of the agent $v_i$ is an endpoint of all of the three edges associated with the positions of agents $e_{j_1}$, $e_{j_2}$, and $e_{j_3}$.  
	
	This now fully describes the reduction. Note that in the definition of $f$ there are $4(x+y)$ type-(iv) and type-(ii) steps and $n + 2m$ type-(iii) and the remaining steps are type-(i). So for each $d\in \{1,\ldots, |V(P)|-1\}$ we indeed have $f(d) = 2^{p_d(n)}+q_d$, where $q_d < d$ and $n^3 \ge p_d(n)\ge n^3 - 4(x+y)n - n - 2m > n^3 - 7n^2$. In addition, all the utilities are one of $0$, $1$, $-1$, $-2$, the ratio of two values of $f$ at different positions multiplied by a small constant. Therefore, each used number can be encoded using at most $\Oh{n^3}$ bits. Since the length of $P$ is also $\Oh{n^3}$, the size of the reduced instance of \IRTDGshort is polynomial in the size of the original instance of PSI.
	
	It is rather straightforward to see that given a solution to the instance $(G,H,\psi)$ of PSI, we can place $g_1$, $g_2$ at the endpoints, anchor agent-pairs at the specified anchor positions, the vertex-selection/edge-selection agent-pairs at the position associated with the vertices and edges selected in $G$ by the injective mapping $\phi$. 
	
	On the other hand, given an allocation of the agents for the reduced instance of \IRTDGshort{}, it is again rather straightforward to observe that each vertex-selection agent-pair $(v_i, w_i)$ is at a position representing the selection of a vertex in $V_i$ and that if $H$ contains an edge between vertices associated with sets $V_{i_1}$ and $V_{i_2}$, then there is a edge-selection agent-pair $(e_i, e'_i)$ responsible for selection of an edge between $V_{i_1}$ and $V_{i_2}$ and the fact that utilities of both $v_{i_1}$ and $v_{i_2}$ are positive guarantees that the edge associated with the allocation of $e_j$ is exactly the edge between the vertex associated with the allocation of $v_{j_1}$ and the vertex associated with the allocation of~$v_{j_2}$. 
	
	Finally, $|N| = 3 + 2 (|V(H)| + |E(H)|)$, where $|E(H)| = \frac{3|V(H)|}{2}$, as $H$ is 3-regular. So $|N| = 5|V(H)| + 3$ and the reduction is parameter-preserving which implies that \IRTDGshort{} is \W[1]-hard even on such restricted instances. In addition, any algorithm for \IRTDGshort{} that runs it time $g(|N|)\cdot |V(P)|^{o(\frac{|N|}{\log |N|})}$ gives an algorithm for PSI that runs in time $g(5|V(H)| + 3)\cdot |V(G)|^{o(\frac{|V(H)|}{\log |V(H)|})}$, which is not possible unless ETH fails. 
\end{proof}

On the other hand, if we additionally restrict the enmity graph, we finally obtain fixed-parameter tractability. Namely, if we parameterize by the number of agents, the topology is a path, and all edges in the enmity graph are oriented towards a single agent, the problem becomes tractable.

\begin{theorem}\label{thm:FPT_path}
	For any distance factor function $\dffn$, if all the edges in the enmity graph are oriented towards one agent $p$ and the topology is a path, then the \IRTDG problem is in \FPT parameterized by the number of agents $\agents$.%
\end{theorem}
\begin{proof}
	Let the graph be a path $P=\{v_1,\ldots,v_n\}$, where $n$ is the number of vertices. 
	As the first step of our algorithm, we set $\assgn(p) = v_n$. Next, we guess the ordering $\pi\colon\agents\setminus\{p\}\to[\numagents-1]$ of the vertices on the path. Now, for every $i\in[\numagents-1]$, we set $\assgn(\pi^{-1}(i)) = v_i$.
	
	For the correctness, we show that if there exists an individually rational solution~$\assgn'$, then there also exists an individually rational solution $\assgn$ where agents in $\agents\setminus\{p\}$ are assigned only to vertices $v_1,\ldots,v_{\numagents-1}$ and $p$ is assigned to~$v_n$.
	
	Let $\pi'$ be the ordering of the agents on the path $P$ with respect to the assignment~$\assgn'$. Additionally, assume that $\pi'(p) \not= \numagents$. We denote by $V_L$ the set of agents that are before the agent~$p$ in the ordering~$\pi'$ and by $V_R$ the set of agents that are after the agent~$p$ in~$\pi'$, respectively. We create a new assignment $\assgn$ such that we assign $p$ to $v_n$, we assign the agent $i$ with the smallest index in the ordering $\pi'$ to a vertex in the distance $\dist_P(\assgn'(p),\assgn'(i))$ from $v_n$, and we assign the remaining agents to a sub-path with the right-most vertex being~$\assgn(i)$ in a way that the higher the index of an agent in $\pi'$ is, the farthest from $v_n$ the agents is. For every agent $i\in V_L\cup V_R$ such that $p$ is its enemy, the utility is either the same or increased since $p$ is in the same distance or even farther and all friends of $i$ are in the same distance or closer. For agents for which $p$ is not an enemy, the utility in $\assgn$ is possibly decreased; however, the utility can never be negative as they do not have any enemy. Consequently, the assignment $\assgn$ is also IR, and we can assume that $p$ is assigned to the vertex $v_n$.
	
	Next, assuming that $p$ is assigned to the vertex $v_n$ in $\assgn'$, we show that if $\assgn'$ is IR, then also $\assgn$ with all agents assigned to vertices $v_{1},\ldots,v_{\numagents}$ is IR. Again, let $\pi'$ be the ordering of the agents on the path. Let $i,j$ be a pair of agents such that $\pi'(i) + 1 = \pi'(j)$ and $\alpha = \dist_P(\assgn'(i),\assgn'(j)) \not= 1$. Additionally, let $\beta = \dist_P(v_1,\assgn'(\pi'^{-1}(1)))$. We define a new assignment $\assgn$ such that for every agent $i\in\agents\setminus\{p\}$ with $v_k = \assgn'(i)$, we set $\assgn(i) = v_{k-\beta+1}$ if $\pi'(i) \leq \pi'(i)$ and $\assgn(i) = v_{k-\alpha-\beta+1}$. By this, for every agent $i$ for which the agent $p$ is an enemy, we either increased the utility or the utility is the same since friends remain in the same distance or are closer, and $p$ can be only farther or in the same distance. For all the remaining agents, we may decrease their utility; however, their utility can never be negative as they do not have enemies. Hence, $\assgn$ is also individually rational, finishing the correctness of the algorithm.
	
	The overall running time is $\Oh{\numagents!\cdot|V(G)|}$, which is clearly in \FPT when parameterized by the number of agents~$\numagents$.
\end{proof}

The algorithm in the previous theorem heavily relies on the special path topology. In our next result, we show that this restriction is necessary for tractability; if we allow for an unrestricted topology, the \IRTDGshort problem again becomes hopelessly intractable.

\begin{theorem}\label{thm:Wc:two_types}
	For every distance factor function $\dffn$, if all edges of the enmity graph are oriented toward one agent, then the \IRTDG problem is \Wc parameterized by the number of agents $\agents$, even if there are only $2$ types of agents and the utility function uses three different values.
\end{theorem}
\begin{proof}
	We show \Whness by a parameterized reduction from the \textsc{Clique} problem, which is known to be \Wc when parameterized by the solution size~$k$~\cite{DowneyF1995}. Let $\mathcal{I} = (H,k)$ be an instance of \textsc{Clique}. We construct an equivalent instance $\mathcal{J}$ of the \IRTDG problem as follows.
	
	The topology $G$ is the graph $H$ with added apex vertex~$c$ with a pendant~$p$. By this tweak, every pair of vertices is now in distance $1$ or $2$. The set of agents $\agents$ consists of $k$ \emph{selection agents} $a_1,\ldots,a_k$ and a single \emph{guard agent} $g$. Next, we define the utilities. For the guard agent, we set $\util_g(a_i) = 0$ for every $i\in[k]$. Next, each selection agent $a_i$ receives negative utility from the guard agent and positive utility from other selection agents. The utilities are set so that the guard agent has to be in distance $2$ and the selection agents have to form a clique. Otherwise, a selection agent with fewer than $k-1$ direct neighbors would have negative utility. Specifically, we set $\util_{a_i}(g) = -\beta$ and $\util_{a_i}(a_j) = \frac{\dffn(2)\cdot\beta}{\dffn(1)\cdot(k-1)}$, where $\beta\in\mathbb{R}_{\geq 0}$ is a fixed constant. As utility functions for selection agents are the same, there are clearly only two types of agents. It is also easy to see that utilities use only $3$ different values, namely $0$, $\beta$, and $\frac{\dffn(2)\cdot\beta}{\dffn(1)\cdot(k-1)}$.
	
	For correctness, let $\mathcal{I}$ be \emph{yes}-instance and $K = \{v_{j_1},\ldots,v_{j_k}\}$ be a clique of size~$k$. We define the assignment $\assgn$ as follows. We set $\assgn(g) = p$, and for each selection agent $a_i$, $i\in[k]$, we set $\assgn(a_i) = v_{j_i}$. Under this assignment $\assgn$, the guard agent is in distance $2$ from every selection agent, and, since $K$ is a clique in $H$, the selection agents are in distance $1$. Consequently, for each selection agent $a_i$, $i\in[k]$, we have
	\begin{align*}
		\util_{a_i}(\assgn) 
		= -\dffn(2)\cdot\beta + \dffn(1)\cdot(k-1)\cdot\frac{\dffn(2)\cdot\beta}{\dffn(1)\cdot(k-1)}
		= -\dffn(2)\cdot\beta + \dffn(2)\cdot\beta = 0.
	\end{align*}
	As the utility of the guard agent is always non-negative, the assignment $\assgn$ is clearly individually rational.
	
	In the opposite direction, let $\mathcal{J}$ be a \emph{yes}-instance and $\assgn$ be individually rational assignment. First, we show that the guard agent must be in distance $2$ from each selection agent.
	
	\begin{claim}
		For every selection agent $a_i$, $i\in[k]$, it holds that $\dist(\assgn(a_i),\assgn(g)) = 2$.
	\end{claim}
	\begin{claimproof}
		For the sake of contradiction, let there exist a selection agent $a_i$ such that $\dist(\assgn(a_i),\assgn(g)) = 1$. Then the utility of $a_i$ is
		\begin{align*}
			\util_{a_i}(\assgn) 
			&= -\dffn(1)\cdot\beta + \sum_{j\in[k]\setminus\{i\}}\dffn(\dist(\assgn(a_i),\assgn(a_j)))\cdot \util_{a_i}(a_j)\\
			&\leq -\dffn(1)\cdot\beta + \dffn(1)\cdot(k-1)\cdot\frac{\dffn(2)\cdot\beta}{\dffn(1)\cdot(k-1)}
			= -\dffn(1)\cdot\beta + \dffn(2)\cdot\beta,
		\end{align*}
		which is clearly negative as $\dffn(1) > \dffn(2)$. This is a contradiction with $\assgn$ being individually rational. As the diameter of the graph $G$ is $2$, the claim follows.
	\end{claimproof}
	
	On the basis of the previous claim, we can assume that $\assgn(g) = p$. Clearly, no selection agent is assigned to the vertex $c$ as $c$ is in distance $1$ from every other vertex. Therefore, if $\assgn(g) \not= p$, then by reassigning $g$ to $p$, we keep all distances and utilities. In the next auxiliary claim, we show that in such an individually rational solution, each selection agent has exactly $k-1$ selection agents in distance $1$.
	
	\begin{claim}\label{cl:selectionAgentsDist}
		For each pair of distinct selection agents $a_i,a_j$, $i,j\in[k]$, we have $\dist(\assgn(a_i),\assgn(a_j)) = 1$.
	\end{claim}
	\begin{claimproof}
		We again prove the claim by contradiction. Assume that there exists a pair of distinct vertices $a_i,a_j$ such that $\dist(\assgn(a_i),\assgn(a_j)) = 2$, and that $\assgn$ is individually rational. Then the utility of $a_i$ is
		\begin{align*}
			\util_{a_i}(\assgn)
			&= -\dffn(2)\cdot\beta + \sum_{j\in[k]\setminus\{i\}} \dffn(\dist(\assgn(a_i),\assgn(a_j))\cdot\util_{a_i}(a_j)\\
			&\leq -\dffn(2)\cdot\beta + \dffn(1)\cdot(k-2)\cdot\util_{a_i}(a_j) + \dffn(2)\cdot\util_{a_i}(a_j)\\
			&= -\dffn(2)\cdot\beta + \util_{a_i}(a_j)\cdot(\dffn(1)\cdot(k-2) + \dffn(2))\\
			&= -\dffn(2)\cdot\beta + \util_{a_i}(a_j)\cdot(\dffn(1)\cdot k - \dffn(1)\cdot2 + \dffn(2))\\
			&= -\dffn(2)\cdot\beta + \frac{\dffn(2)\cdot\beta}{\dffn(1)\cdot(k-1)}\cdot(\dffn(1)\cdot k - \dffn(1)\cdot2 + \dffn(2))\\
			&= -\dffn(2)\cdot\beta + \dffn(2)\cdot\beta\cdot\left(\frac{\dffn(1)\cdot k - 2\cdot\dffn(1) + \dffn(2)}{\dffn(1)\cdot(k-1)}\right)\\
			&= -\dffn(2)\cdot\beta + \dffn(2)\cdot\beta\cdot\left(\frac{\dffn(1)(k-1) - \varepsilon}{\dffn(1)(k-1)}\right),
		\end{align*}
		where $\varepsilon > 0$ since $\dffn(1) > \dffn(2)$. Consequently, we see that $\util_{a_i}(\assgn)$ is strictly smaller than $0$. That is, such an assignment~$\assgn$ is not individually rational, which is a contradiction. Hence, all solution agents are necessarily at a distance $1$.
	\end{claimproof}
	
	Now, we construct a set $K\subseteq V(H)$ and show that $K$ is a solution for $\mathcal{I}$. We set $K = \{\assgn(a_i)\mid i\in[k]\}$. Clearly, the set is of size $k$. Suppose that the vertices of $K$ do not form a clique in $H$. Then there exists a pair of distinct vertices $u,v\in K$ such that $\{u,v\}\not\in E(H)$. Recall that the set $K$ was constructed so that there are two selection agents $a_i,a_j$ such that $\assgn(a_i) = u$ and $\assgn(a_j) = v$. By \Cref{cl:selectionAgentsDist}, they have to be in distance $1$; therefore, there exists an edge $\{\assgn(a_i),\assgn(a_j)\}\in E(G)$. Moreover, we already argued that no selection agent is assigned to the vertex $c$ (and, consequently, also none of them is assigned to the vertex $p$). Since each edge that is in $E(G)$ and is not in $E(H)$ contains either vertex $c$ or $p$, it follows that the edge $\{u,v\}$ is necessarily also in $E(H)$. Therefore, such a pair of vertices can never occur and $K$ necessarily forms a clique in $H$.
	
	To conclude, it is easy to see that the reduction can be performed in polynomial time and, moreover, we have $\numagents = k+1$. That is, the reduction is indeed a parameterized reduction and the theorem follows.
\end{proof}

In the following result, we restrict the distance factor function. More specifically, the distance factor function is defined as a monotically decreasing function; however, we now assume that $f$ is monotically decreasing up to some fixed position $d$ and for all $d' > d$, the distance factor functions return zero. This assumption, despite that it deviates from the standard definition, is well motivated by practice: even though grumpy-John and prickly-Jack would fight if they are close to each other, it is unlikely that this circumstance occurs when they are assigned to completely different rooms. That is, it is natural to assume that the impact of agents that are very far away from each other is negligible. We call such a special function~$f$ a \emph{bounded distance factor function} and~$d$ its boundedness level. Note that a similar assumption is common in many similar scenarios such as social distance games~\cite{BranzeiL2011,KaklamanisKP2018,FlamminiKOV2020,BalliuFMO2022,GanianHKRSO2023}.

With the assumption of a bounded distance factor function, we give an \FPT algorithm for parameterization by the number of agents and the twin-width of the topology combined. The twin-width is a recently introduced structural parameter~\cite{BonnetKTW22} that is becoming very successful in the development of fixed-parameter algorithms for several reasons; first, it possesses a very convenient graph decomposition that can be exploited in the design of dynamic programming algorithms~\cite{BonnetGKTW2021}, and it generalizes many widely-studied structural graph parameters, such as the celebrated tree-width~\cite{BonnetD2023}. It is worth noting that in our result, we assume that the twin-width decomposition of the optimal width is given as part of the input. Although the computation of such an optimal decomposition is already \NPh for graphs of twin-width $4$~\cite{BergeBD2022}, there exist efficient solvers that compute optimal decompositions for practical instances in reasonable time~\cite{SchidlerS2022,SchidlerS2023}.

\begin{theorem}\label{thm:FPT_cw}
	If the distance factor function is bounded, the \IRTDG problem is in \FPT parameterized by the number of agents $\numagents$ and the twin-width of the topology $\tww$ combined, assuming the twin-width decomposition for~$G$ is provided in the input.
\end{theorem}
\begin{proof}[Proof Sketch.]
	It is known~\cite{BonnetKTW22} that First-Order (FO) model checking can be decided on graphs of bounded twin-width.
	That is, there exists an algorithm that, given an FO formula~$\varphi$ with~$k$ and a decomposition of a graph~$G$ of width~$\tww$, decides in $f(\tww,k)\cdot n^\Oh{1}$ time whether~$\varphi$ holds for~$G$ (where $f$ is a computable function).
	
	It is not hard to see that if the distance factor function is bounded by~$d$, i.e., we have $f(d) = f(d+1) = \cdots = 0$, then the solution only depends on ``small'' distances.
	For each~$p \in [d]$ we build a formula $\varphi_p(v_a, v_b)$ that, given two vertices~$v_a$ and~$v_b$ to which we assign agents $a,b \in N$, holds if the two vertices are in distance at most~$p$.
	Such a formula is, e.g.,
	\begin{multline*}
		\varphi_p(v_a, v_b) = \exists v^1\,\exists v^2\, \cdots\, \exists v^p\colon (v^1 = v_a) \land (v^p = v_b) \land \bigwedge_{1 \le i < j \le p} v^i \neq v^j \bigwedge_{i = 1}^{p-1} \{v^i, v^{i+1} \} \in E \,.
	\end{multline*}
	Now, if we have $\varphi_p(v_a, v_b)$ for all $p \in [d]$, we know the distance between the two agents exactly (it equals to the smallest $p$ such a formula holds for).
	In the rest of the proof, we will deal with the complexity the distance function brings to the problem.
	
	To overcome this, we observe that we can precompute for each agent~$a$, a set of agents $X \subseteq N \setminus \{a\}$, and a function $\delta \colon X \in [d]$, whether $a$ has positive utility if the set $X$ is their $d$-neighborhood and their distances in the assignment are~$\delta$.
	Note that all of this amounts to a number bounded by a function dependent solely on~$\numagents$ and $d$.
	If we then collect all such configurations into the set of valid solution parts~$\mathcal{V}(a) = \{ (X, \delta) \}$, we can use these in the formula for deciding the problem.
	
	Overall we are looking for an assignment, i.e., distinct vertices $v_a$ for $a \in N$ so that for each agent at least one valid solution part holds~$(X, \delta)$.
	It is clear that there is at most one valid solution part and therefore this yields exactly one, as needed.
\end{proof}

To conclude, we return to the standard definition of the distance factor function and give one more positive result. Specifically, we show that the \IRTDGshort problem is in \FPT when parameterized by the number of agents and the shrub-depth of the topology combined. Shrub-depth~\cite{GanianHNOM19} is a structural parameter that can be seen as an equivalent of the famous tree-depth for dense graphs. Similarly to tree-depth, graphs of bounded shrub-depth cannot contain long induced paths; hence, their diameter is bounded. Therefore, the algorithm for shrub-depth parameterization follows by a direct application of \Cref{thm:FPT_cw}; the bounded shrub-depth graphs also admit the bounded twin-width and the boundedness of the distance factor function follows from the bounded-diameter property.

\begin{corollary}\label{thm:FPT_sd}
	For every distance factor function $\dffn$, the \IRTDG problem is in \FPT parameterized by the shrub-depth of the topology and the number of agents $\numagents$ combined.
\end{corollary}

Note that shrub-depth, unlike twin-width, can be computed in \FPT time~\cite{GajarskyK2020}. Moreover, shrub-depth is a generalization of the tree-depth. Hence, we see that \IRTDGshort is in \FPT also when parameterized by the number of agents and the tree-depth and consequently also for the combined parameter the number of agents and the vertex-cover number.

\section{Conclusions}
\label{sec:conclusions}
This paper studied the complexity of finding individually rational solutions in topological distance games, which is arguably one of the most fundamental stability notions. 
Albeit this class of games can capture a plethora of models, its versatility comes with the drawback of strong intractability results even for very restricted cases, at least from the theoretical point of view. 
As our results reveal, individual rationality is hard to be assured even if someone resorts to the parameterized complexity regime. However, our results do not imply parameterized-complexity hardness for jump stability. We strongly believe that this avenue deserves further study.

At a different dimension, our results indicate that following a worst-case point of view is not sufficient for tractability. This makes someone wonder, whether there exist some natural values for the parameters of the model that ensure IR in practice. If not, does individual rationality even exist?

\section*{Acknowledgments}

This work was co-funded by the European Union under the project Robotics and advanced industrial production (reg. no. CZ.02.01.01/00/22\_008/0004590).
Argyrios Deligkas acknowledges the support of the EPSRC grant EP/X039862/1.
Šimon Schierreich acknowledges the additional support of the Grant Agency of the Czech Technical University in Prague, grant No. SGS23/205/OHK3/3T/18.

\bibliographystyle{plain}
\bibliography{references}

\appendix
	
\end{document}